\relax
\documentclass[letterpaper]{article} 
\usepackage{url}  
\usepackage{graphicx}  
\usepackage[margin=1in,headheight=25pt]{geometry}
\usepackage{amsmath}
\usepackage{amssymb}
\usepackage{bm}
\usepackage{amsthm}
\usepackage{algorithmicx}
\usepackage{subcaption}
\usepackage{natbib}
\newtheorem{theorem}{Theorem}

\newtheorem{lemma}{Lemma}

\theoremstyle{definition}

\DeclareMathOperator*{\E}{\mathbb{E}}
\newcommand{\e}{\left(1 - \frac{1}{e}\right)}
\newcommand{\polym}{\mathcal{P}}
\newcommand*\Let[2]{\State #1 $\gets$ #2}
\usepackage{algpseudocode,algorithm,algorithmicx}

\newcommand{\shortcite}{\citeyearpar}
\renewcommand{\cite}{\citep}

\begin{document}
\title{Equilibrium Computation and Robust Optimization in Zero Sum Games with Submodular Structure}
\author{Bryan Wilder\\Department of Computer Science and Center for Artificial Intelligence in Society\\University of Southern California\\bwilder@usc.edu
}
\date{}
\maketitle

\begin{abstract}
	We define a class of zero-sum games with combinatorial structure, where the best response problem of one player is to maximize a submodular function. For example, this class includes security games played on networks, as well as the problem of robustly optimizing a submodular function over the worst case from a set of scenarios. The challenge in computing equilibria is that both players' strategy spaces can be exponentially large. Accordingly, previous algorithms have worst-case \emph{exponential} runtime and indeed fail to scale up on practical instances. We provide a pseudopolynomial-time algorithm which obtains a guaranteed $(1 - 1/e)^2$-approximate mixed strategy for the maximizing player. Our algorithm only requires access to a weakened version of a best response oracle for the minimizing player which runs in polynomial time. Experimental results for network security games and a robust budget allocation problem confirm that our algorithm delivers near-optimal solutions and scales to much larger instances than was previously possible.
\end{abstract}

\section{Introduction}


Submodular functions are ubiquitous due to wide-spread applications ranging from machine learning, to viral marketing, to mechanism design. Intuitively, submodularity captures diminishing returns (formalized later). In this paper, we use techniques rooted in submodular optimization to solve previously intractable zero-sum games. We then show how to instantiate our algorithm for two specific games, including the robust optimization of a submodular objective.

As an example, consider the network security game introduced by Tsai et al. \shortcite{tsai2010urban}. A defender can place checkpoints on $k$ edges of a graph. An attacker aims to travel from a source node to any one of several targets without being intercepted. Each player has an exponential number of strategies since the defender may choose any set of $k$ edges and the attacker may choose any path. Hence, previous approaches to computing the optimal defender strategy were either heuristics with no approximation guarantee, or else provided guarantees but ran in worst-case exponential time \cite{jain2011double,iwashita2016simplifying}. 

However, this game has useful structure. The defender's best response to any attacker mixed strategy is to select the edges which are most likely to intersect the attacker's chosen path. Computing this set is a submodular optimization problem \cite{jain2013security}. We give a general algorithm for computing approximate minimax equilibria in zero-sum games where the maximizing player's best response problem is a monotone submodular function. Our algorithm obtains a $(1 - \frac{1}{e})^2-$approximation (modulo an additive loss of $\epsilon$) to the maximizing player's minimax strategy. This algorithm runs in pseudopolynomial time \emph{even when both action spaces are exponentially large} given access to a weakened form of a best response oracle for the adversary. Pseudopolynomial means that the runtime bound depends polynomially on largest value of any single item (which we expect to be a constant for most cases of interest). Our algorithm approximately solves a non-convex, non-smooth continuous extension of the problem and then rounds the solution back to a pure strategy in a randomized fashion. To our knowledge, no subexponential algorithm was previously known for this problem with exponentially large strategy spaces. Our framework has a wide range of applications, corresponding to the ubiquitous presence of submodular functions in artificial intelligence and algorithm design (see Krause and Golovin \shortcite{krause2014submodular} for a survey).

One prominent class of applications is \emph{robust submodular optimization}. A decision maker is faced with a set of submodular objectives $f_1...f_m$. They do not know which objective is the true one, and so would like to find a decision maximizing $\min_i f_i$. Robust submodular optimization has many applications because uncertainty is so often present in decision-making. We start by studying the randomized version of this problem, where the decision maker may select a distribution over actions such that the worst case \emph{expected} performance is maximized \cite{krause2011randomized,chen2017robust,wilder2017uncharted}. This is equivalent to computing the minimax equilibrium for a game where one player has a submodular best response. Our techniques for solving such games also yield an algorithm for the deterministic robust optimization problem, where the decision maker must commit to a single action. Specifically, we obtain bicriteria approximation guarantees analogous to previous work \cite{krause2008robust} under significantly more general conditions.

We make three contributions. First, we define the class of \emph{submodular best response} (SBR) games, which includes the above examples. Second, we introduce the EQUATOR algorithm to compute approximate equilibrium strategies for the maximizing player. Third, we give example applications of our framework to problems with no previously known approximation algorithms. We start out by showing that network security games \cite{tsai2010urban} can be approximately solved using EQUATOR. We then introduce and solve the robust version of a classical submodular optimization problem: robust maximization of a coverage function (which includes well-known applications such as budget allocation and sensor placement). Finally, we experimentally validate our approach for network security games and robust budget allocation. We find that EQUATOR produces near-optimal solutions and easily scales to instances that are too large for previous algorithms to handle.

\section{Problem description}

\textbf{Formulation: } Let $X$ be a set of items with $|X| = n$. A function $f: 2^X \to R$ is submodular if for any $A \subseteq B$ and $i \in X\setminus B$, $f(A \cup \{i\}) -f(A) \geq f(B \cup \{i\}) - f(B)$. We restrict our attention to functions that are $\emph{monotone}$, i.e., $f(A \cup \{i\}) - f(A) \geq 0$ for all $i \in X, A \subset X$. Without loss of generality, we assume that $f(\emptyset) = 0$ and hence $f(S) \geq 0 \, \forall S$. Let $\mathcal{F} = \{f_1...f_m\}$ be a finite set of submodular functions on the ground set $X$. $m$ may be exponentially large. Let $\Delta(S)$ denote the set of probability distributions over the elements of any set $S$. Oftentimes, we will work with independent distributions over $X$, which can be fully specified by a vector $\bm{x} \in R^n_+$. $x_i$ gives the marginal probability that item $i$ is chosen. Denote by $p^I_{\bm{x}}$ the independent distribution with marginals $\bm{x}$.  Let $\mathcal{I}$ be a collection of subsets of $X$. For instance, we could have $\mathcal{I} = \{S \subseteq X : |S| \leq k\}$. We would like to find a minimax equilibrium of the game where the maximizing player's pure strategies are the subsets in $\mathcal{I}$, and the minimizing player's pure strategies are the functions in $\mathcal{F}$. The payoff to the strategies $S \in \mathcal{I}$ and $f_i \in \mathcal{F}$ is $f_i(S)$. We call a game in this form a \emph{submodular best response} (SBR) game. For the maximizing player, computing the minimax equilibrium is equivalent to solving
\begin{align}\label{problem:randomized}
\max_{p \in \Delta(\mathcal{I})} \min_{f \in \mathcal{F}} \E_{S \sim p}[f(S)]
\end{align}

where $S \sim p$ denotes that $S$ is distributed according to $p$. 

\textbf{Example: network security games. } To make the setting more concrete, we now introduce one of our example domains, the network security game of Tsai et al.\ \shortcite{tsai2010urban}. There is a graph $G = (V, E)$. There is a source vertex $s$ (which may be a supersource connected to multiple real sources) and a set of targets $T$. An attacker wishes to traverse the network starting from the source and attack a target. Each target $t_j$ has a value $\tau_j$. The attacker picks a $s-t_j$ path for some $t_j \in T$. The defender attempts to catch the attacker by protecting edges of the network. The defender may select any $k$ edges, and the attacker is caught if any of these edges lies on the chosen path. We use the normalized utilities defined by Jain et al.\ \shortcite{jain2013security}, which give the defender utility $\tau_j > 0$ if an attack on $t_j$ is intercepted and 0 if the attack succeeds.  Thus, each path $P$ from $s$ to $t_j$ for the attacker induces an objective function $f_P$ for the defender: for any set of edges $S$, $f_P(S) = \tau_j$ if $S \cap P \not= \emptyset$, otherwise $f_P(S) = 0$. $f_P$ is easily seen to be submodular \cite{jain2013security}. Hence, we have a SBR game with $\mathcal{I} = \{S \subseteq E: |S| \leq k\}$ and $\mathcal{F} = \{f_P: P \text{ is a path from } S \text{ to } T\}$.


\textbf{Allowable pure strategy sets: } Our running example is when the pure strategies $\mathcal{I}$ of the maximizing player are all size $k$ subsets: $\mathcal{I} = \{S \subseteq X : |S| \leq k\}$. In general, our algorithm works when $\mathcal{I}$ is any matroid; this example is called the \emph{uniform} matroid. We refer to \cite{korte2012combinatorial} for more details on matroids. Here, we just note that matroids are a class of well-behaved constraint structures which are of great interest in combinatorial optimization. A useful fact is that any linear objective can be exactly optimized over a matroid by the greedy algorithm. For instance, consider the above uniform matroid. If each element $j$ has a weight $w_j$, the highest weighted set of size $k$ is obtained simply by taking the $k$ items with highest individual weights. Let $k = \max_{S \in \mathcal{I}} |S|$ be the size of the largest pure strategy. E.g., in network security games $k$ is the number of defender resources. In general, $k$ is the rank of the matroid. 

We now introduce some notation for the continuous extension of the problem. Let $\bm{1}_S$ be the indicator vector of the set $S$ (i.e., an $n$-dimensional vector with 1 in the entries of elements that are in $S$ and 0 elsewhere). Let $\polym$ be the convex hull of $\{\bm{1}_S : S \in \mathcal{I}\}$. Note that $\polym$ is a polytope. 


\textbf{Best response oracles: } A best response oracle for one player is a subroutine which computes the pure strategy with highest expected utility against a mixed strategy for the other player. We assume that an oracle is available for the minimzing player. However, we require only a weaker oracle, which we call an \emph{best response to independent distributions oracle} (BRI). A BRI oracle is only required to compute a best response to mixed strategies which are independent distributions, represented as the marginal probability that each item in $X$ appears. Given a vector $\bm{x} \in R_+^n$, where $x_i$ is the probability that element $i \in X$ is chosen, a BRI oracle computes $\arg\min_{f_i \in \mathcal{F}} \E_{S \sim p^I_{\bm{x}}}[f_i(S)]$. We use $S \sim \bm{x}$ to denote that $S$ is drawn from the independent distribution with marginals $\bm{x}$. As we will see later, sometimes a BRI oracle is readily available even when the full best response is NP-hard. 

\textbf{Robust optimization setting:} One prominent application of SBR games is robust submodular optimization. Robust optimization models decision making under uncertainty by specifying that the objective is not known exactly. Instead, it lies within an uncertainty set $\mathcal{U}$ which represents the possibilities that are consistent with our prior information. Our aim is to perform well in the worst case over all objectives in $\mathcal{U}$. We can view this as a zero sum game, where the decision maker chooses a distribution over actions and nature adversarially chooses the true objective from $\mathcal{U}$. A great deal of recent work has been devoted to the setting of randomized actions, both because randomization can improve worst-case expected utility \cite{delage2016dice}, and because the randomized version often has much better computational properties \cite{krause2011randomized,orlin2016Robust}. Randomized decisions also naturally fit a problem setting where the decision maker will take several actions and wants to maximize their total reward. Any single action might perform badly in the worst case; drawing the actions from a distribution allows the decision maker to hedge their bets and perform better overall.

\section{Previous work}

We discuss related work in two areas. First, solving zero-sum games with exponentially large strategy sets. Efficient algorithms are known only for limited special cases. One approach is to represent the strategies in a lower dimensional space (the space of marginals). We elaborate more below since our algorithm uses this approach. For now, we just note that previous work \cite{ahmadinejad2016duels,xu2016mysteries,chan2016multilinear} requires that the payoffs be \emph{linear} in the lower dimensional space. Linearity is a very restrictive assumption; ours is the first algorithm which extends the marginal-based approach to general submodular functions. This requires entirely different techniques. 

In practice, large zero sum games are often solved via the \emph{double oracle} algorithm \cite{mcMahan2003planning,bosansky2014exact,bosansky2015combining,halvorson2009multi}. Double oracle starts with each player restricted to only a small number of arbitrarily chosen pure strategies and repeatedly adds a new strategy for each player until an equilibrium is reached. The new strategies are chosen to be each player's best response to the other's current mixed strategy. This technique is appealing when equilibria have sparse support, and so only a few iterations are needed. However, it is easy to give examples where \emph{every} pure strategy lies in the support of the equilibrium, so double oracle will require exponential runtime. Our algorithm runs in guaranteed polynomial time.

Second, we give more background on robust submodular optimization. Krause et al.\ \shortcite{krause2008robust} introduced the problem of maximizing the minimum of submodular functions, which corresponds to Problem \ref{problem:randomized} with the maximizing player restricted to pure strategies. They show that the problem is inapproximable unless P = NP. They then relax the problem by allowing the algorithm to exceed the budget constraint (a bicriteria guarantee). Our primary focus is on the \emph{randomized} setting, where the algorithm respects the budget constraint but chooses a distribution over actions instead of a pure strategy. This randomized variant was studied by Wilder et al. \shortcite{wilder2017uncharted} for the special case of influence maximization. Krause et al.\ \shortcite{krause2011randomized} and Chen et al.\ \shortcite{chen2017robust} studied general submodular functions using very similar techniques: both iterate dynamics where the adversary plays a no-regret learning algorithm and the decision maker plays a greedy best response. This algorithm maintains a variable for every function in $\mathcal{F}$ and so is only computationally tractable when $\mathcal{F}$ is small. By contrast, we deal with the setting where $\mathcal{F}$ is exponentially large. However, we lose an extra factor of $(1 - 1/e)$ in the approximation ratio. 

We also extend our algorithm to obtain bicriteria guarantees for the deterministic robust submodular optimization problem (where we select a single feasible set). Our guarantees apply under significantly more general conditions than those of Krause et al. \shortcite{krause2008robust} but have weaker approximation guarantee; details can be found in the discussion after Theorem \ref{theorem:bicriteria}. 

\section{Preliminaries}

We now introduce techniques our algorithm builds on.

\textbf{Multilinear extension: } We can view a set function $f$ as being defined on the vertices of the hypercube $\{0,1\}^{n}$. Each vertex is the indicator vector of a set. A useful paradigm for submodular optimization is to extend $f$ to a continuous function over $[0,1]^n$ which agrees with $f$ at the vertices. The \emph{multilinear extension} $F$ is defined as 
\begin{align*}
F(\bm{x}) = \sum_{S \subseteq X} f(S) \prod_{j \in S} x_j \prod_{j \not\in S} 1 - x_j.
\end{align*}
Equivalently, $F(\bm{x}) = \E_{S \sim \bm{x}}[f(S)]$. That is, $F(\bm{x})$ is the expected value of $f$ on sets drawn from the independent distribution with marginals $\bm{x}$. $F$ can be evaluated using random sampling \cite{calinescu2011maximizing} or in closed form for special cases \cite{iyer2014monotone}. Note that for any set $S$ and its indicator vector $\bm{1}_S$, $F(\bm{1}_S) = f(S)$. One crucial property of $F$ is \emph{up-concavity} \cite{calinescu2011maximizing}. That is, $F$ is concave along any direction $\bm{u} \succeq \bm{0}$ (where $\succeq$ denotes element-wise comparison). Formally, a function $F$ is \emph{up-concave} if for any $\bm{x}$ and any $\bm{u} \succeq 0$, $F(\bm{x}  +\xi \bm{u})$ is concave as a function of $\xi$. 

\textbf{Correlation gap: } A useful property of submodular functions is that little is lost by optimizing only over independent distributions. Agrawal et al.\ \shortcite{agrawal2010correlation} introduced the concept of the correlation gap, which is the maximum ratio between the expectation of a function over an independent distribution and its expectation over a (potentially correlated) distribution with the same marginals. Let $D(\bm{x})$ be the set of distributions with marginals $\bm{x}$. The correlation gap $\kappa(f)$ of a function $f$ is defined as 

\begin{align*}
\kappa(f) = \max_{\bm{x} \in [0,1]^n}\max_{p \in D(\bm{x})} \frac{\E_{S \sim p}[f(S)]}{\E_{S \sim p_{\bm{x}}^I}[f(S)]}.
\end{align*}

For any submodular function $\kappa \leq \frac{e}{e - 1}$. This says that, up to a loss of a factor $1  - 1/e$, we can restrict ourselves to independent distributions when solving Problem \ref{problem:randomized}.

\textbf{Swap rounding: } Swap rounding is an algorithm developed by Chekuri et al.\ \shortcite{chekuri2010dependent} to round a fractional point in a matroid polytope to an integral point. We will use swap rounding to convert the fractional point obtained from the continuous optimization problem to a distribution over pure strategies. Swap rounding takes as input a representation of a point $\bm{x} \in \polym$ as a convex combination of pure strategies. It then merges these sets together in a randomized fashion until only one remains. For \emph{any} submodular function $f$ and its multilinear extension $F$, the random set $R$ satisfies $\E[f(R)] \geq F(\bm{x})$. I.e., swap rounding only increases the value of any submodular function in expectation.

\section{Algorithm for SBR games}

In this section, we introduce the EQUATOR (\emph{EQUilibrium via stochAsTic frank-wOlfe and Rounding}) algorithm for computing approximate equilibrium strategies for the maximizing player in SBR games. Since the pure strategy sets can be exponentially large, it is unclear what it even means to compute an equilibrium: representing a mixed strategy may require exponential space. Our solution to this dilemma is to show how to efficiently \emph{sample} pure strategies from an approximate equilibrium mixed strategy. This suffices for the maximizing player to implement their strategy. Alternatively, we can build an approximate mixed strategy with sparse support by drawing a polynomial number of samples and outputing the uniform distribution over the samples. In order to generate these samples, EQUATOR first solves a continuous optimization problem, which we now describe.

\textbf{The marginal space: }A common meta-strategy for solving games with exponentially large strategy sets is to work in the lower-dimensional space of \emph{marginals}. I.e., we keep track of only the marginal probability that each element in the ground set is chosen. To illustrate this, let $p$ be a distribution over the pure strategies $\mathcal{I}$, and $\bm{x} \in \polym$ denote a vector giving the marginal probability of selecting each element of $X$ in a set drawn according to $p$. Note that $\bm{x}$ is $n$-dimensional while $p$ could have dimension up to $2^n$.  Previous work has used marginals for linear objectives. A linear function with weights $w$ satisfies $\E_{S \sim p}\left[\sum_{j \in S} w_j\right] = \sum_{j = 1}^n w_j \text{Pr}[j \in S] = \sum_{j = 1}^n w_j x_j$, so keeping track of only the marginal probabilities $\bm{x}$ is sufficient for exact optimization. However, submodular functions do not in general satisfy this property: the utilities will depend on the full distribution $p$, not just the marginals $\bm{x}$. We will treat a given marginal vector $\bm{x}$ as representing an independent distribution where each $j$ is present with probability $x_j$ (i.e., $\bm{x}$ compactly represents the full distribution $p^I_{\bm{x}}$). The expected value of $\bm{x}$ under any submodular function is exactly given by its multilinear extension, which is a continuous function. 

\textbf{Continuous extension: }Let $G = \min_i F_i$ be the pointwise minimum of the multilinear extensions of the functions in $\mathcal{F}$. Note that for any marginal $\bm{x}$, $G(\bm{x})$ is exactly the objective value of $p^I_{\bm{x}}$ for Problem \ref{problem:randomized}. Hence, optimizing $G$ over all $\bm{x} \in \mathcal{P}$ is equivalent to solving Problem \ref{problem:randomized} restricted to independent distributions. Via the correlation gap, this restriction only loses a factor $(1 - 1/e)$: if the optimal full distribution is $p_{OPT}$, then the independent distribution with the same marginals as $p_{OPT}$ has at least $(1 - 1/e)$ of of $p_{OPT}$'s value under any submodular function. Previous algorithms \cite{calinescu2011maximizing,bian2017guaranteed} for optimizing up-concave functions like $G$ do not apply because $G$ is nonsmooth (see below). We introduce a novel Stochastic Frank-Wolfe algorithm which smooths the objective with random noise. Its runtime does not depend directly on $|\mathcal{F}|$ at all; it only uses BRI calls.

\textbf{Rounding: }Once we have solved the continuous problem, we need a way of mapping the resulting marginal vector $\bm{x}$ to a distribution over the pure strategies $\mathcal{I}$. Notice that if we simply sample items independently according to $\bm{x}$, we might end up with an invalid set. For instance, in the uniform matroid which requires $|S| \leq k$, an independent draw could result in more than $k$ items even if $\sum_i x_i \leq k$. Hence, we sample pure strategies by running the swap rounding algorithm on $\bm{x}$. In order to implement the maximizing player's equilibrium strategy, it suffices to simply draw a sample whenever a decision is required. If a full description of the mixed strategy is desired, we show that it is sufficient to draw $\Theta\left(\frac{1}{\epsilon^2}(\log |\mathcal{F}| + \log \frac{1}{\delta})\right)$ independent samples via swap rounding and return the uniform distribution over the sampled pure strategies.

To sum up, our strategy is as follows. First, solve the continuous optimization problem to obtain marginal vector $\bm{x}$. Second, draw sampled pure strategies by running randomized swap rounding on $\bm{x}$. 

\subsection{Solving the continuous problem}
The linchpin of our algorithmic strategy is solving the optimization problem $\max_{\bm{x} \in \polym} G(\bm{x})$. In this section, we provide the ingredients to do so. 

\textbf{Properties of $G$: }We set the stage with four important properties of $G$ (proofs are given in the supplement). First, while $G$ is not in general concave, it is \emph{up-concave}:

\begin{lemma} \label{lemma:min-upconcave}
	If $F_1...F_m$ are up-concave functions, then $G = \min_i F_i$ is up-concave as well.
\end{lemma} 
The proof is similar to the proof that the minimum of concave functions is concave. Up-concavity of $G$ is the crucial property that enables efficient optimization.
\begin{algorithm}
	\caption{EQUATOR$(BRI, FO, LO, u, c, K, r)$} \label{alg:equator}
	\begin{algorithmic}[1] 
		\Let{$\bm{x}^0$}{$u\bm{1}$} 
		\State //Stochastic Frank-Wolfe algorithm
		\For{$\ell = 1...K$}
		\For{$t = 1...c$}
		\State Draw $\bm{z} \sim \mu(u)$
		\Let{$F_t$}{BRI($\bm{x}^{\ell - 1} + \bm{z}$)}
		\Let{$\tilde{\nabla}^\ell_t$}{$FO(F_t, \bm{x}^{\ell - 1} + \bm{z})$}
		\EndFor
		\Let{$\tilde{\nabla}^\ell$}{$\frac{1}{c}\sum_{t = 1}^c \tilde{\nabla}^\ell_t$}
		\Let{$\bm{v}^\ell$}{$LO(\tilde{\nabla}^\ell)$}
		\Let{$\bm{x}^{\ell}$}{$\bm{x}^{\ell - 1} + \frac{1}{K} \bm{v}^\ell$}
		\EndFor
		\Let{$\bm{x}_{final}$}{$\bm{x}^K - u\bm{1}$}
		\State //Sample from equilibrium mixed strategy
		\State Return $r$ samples of SwapRound($\bm{x}_{final}$)
	\end{algorithmic}
\end{algorithm}

Second, $G$ is Lipschitz. Specifically, let $M = \max_{i, j}f_i(\{j\})$ be the maximum value of any single item. It can be shown that $||\nabla F_i||_\infty \leq M  \,\,\forall i$ since (intuitively), the gradient of $F_i$ is related to the marginal gain of items under $f_i$. From this we derive

\begin{lemma}\label{lemma:lipschitz}
	$G$ is $M$-Lipschitz in the $\ell_1$ norm.
\end{lemma}

Third, $G$ is not smooth. For instance, it is not even differentiable at points where the minimizing function is not unique. This complicates the problem of optimizing $G$ and renders earlier algorithms inapplicable.

Fourth, at any point $\bm{x}$ where the minimizing function $F_i$ is unique, $\nabla G(\bm{x}) = \nabla F_i (\bm{x})$. Hence, we can compute $\nabla G(\bm{x})$ by calling the BRI to find $F_i$, and then computing $\nabla F_i(\bm{x})$. In general, $\nabla F_i(\bm{x})$ can be computed by random sampling \cite{calinescu2011maximizing}, and closed forms are known for particular cases \cite{iyer2014monotone}.

\textbf{Randomized smoothing:} We will solve the continuous problem $\max_{\bm{x} \in \polym} G(\bm{x})$. Known strategies for optimizing up-concave functions \cite{bian2017guaranteed} rely crucially on $G$ being smooth. Specifically, $\nabla G$ must be Lipschitz continuous. Unfortunately, $G$ is not even differentiable everywhere. Even between two points $\bm{x}$ and $\bm{y}$ where $G$ is differentiable, $\nabla G(\bm{x})$ and $\nabla G(\bm{y})$ can be arbitrarily far apart if $\arg \min_i F_i (\bm{x}) \not = \arg \min_i F_i (\bm{y})$. No previous work addresses nonsmooth optimization of an up-concave function. 

To resolve this issue, we use a carefully calibrated amount of random noise to smooth the objective. Let $\mu(u)$ be the uniform distribution over the $\ell_\infty$ ball of radius $u$. We define the smoothed objective $G_\mu(\bm{x}) = \E_{\bm{z} \sim \mu(u)}\left[G(\bm{x} + \bm{z})\right]$ which averages over the region around $\bm{x}$. This (and similar) techniques have been studied in the context of convex optimization \cite{duchi2012randomized}. We show that $G_\mu$ is a good smooth approximator of $G$. 

\begin{lemma} \label{lemma:smooth-approx}
	$G_\mu$ has the following properties:
	\begin{itemize}
		\item $G_\mu$ is up-concave.
		\item $|G_\mu(\bm{x}) - G(\bm{x})| \leq \frac{Mnu}{2} \quad \forall \bm{x}$.
		\item $G_\mu$ is differentiable, with $\nabla G_\mu(x) = \E[\nabla G(\bm{x} + \bm{z})]$.
		\item $\nabla G_\mu$ is $\frac{M}{\mu}-$Lipschitz continuous in the $\ell_1$ norm.
	\end{itemize}
\end{lemma}

Hence, we can use $G_\mu$ as a better-behaved proxy for $G$ since it is both smooth and close to $G$ everywhere in the domain. The main challenge is that $G_\mu$ and its gradients are not available in closed form. Accordingly, we randomly sample values of the perturbation $\bm{z}$ and average over the value of $G$ (or its gradient) at these sampled points. 

\subsection{Stochastic Frank-Wolfe algorithm (SFW)}

We propose the SFW algorithm (Algorithm \ref{alg:equator}) to optimize $G_\mu$. SFW generates a series of feasible points $\bm{x}^0...\bm{x}^K$, where $K$ is the number of iterations. Each point is generated from the last via two steps. First, SFW estimates the gradient of $G_\mu$. Second, it takes a step towards the point in $\mathcal{P}$ which is furthest in the direction of the gradient. To carry out these steps, SFW requires three oracles. First, a linear optimization oracle $LO$ which, given an objective $\bm{w}$, returns $\arg\max_{\bm{v} \in \polym} \bm{w}^\top \bm{v}$. In the context of our problem, $LO$ outputs the indicator vector of the set $S \in \mathcal{I}$ which maximizes the linear objective $\bm{w}$. $S$ can be efficiently found via the greedy algorithm. The other two oracles concern gradient evaluation. One is the BRI oracle discussed earlier. The other is a stochastic first-order oracle $FO$ which, for any function $F_i$ and point $\bm{x}$, returns an unbiased estimate of $\nabla F_i(\bm{x})$.

The algorithm starts at $\bm{x}^0 = \bm{0}$. At each iteration $\ell$, it averages over $c$ calls to $FO$ to compute a stochastic approximation $\tilde{\nabla}^\ell$ to $\nabla G_\mu(\bm{x}^{\ell-1})$  (Lines 4-9). For each call, it draws a random perturbation $\bm{z} \sim \mu(u)$ and uses the BRI to find the minimizing $F$ at $\bm{x}^{\ell - 1} + \bm{z}$. It then queries $FO$ for an estimate of $\nabla F(\bm{x}^{\ell - 1} + \bm{z})$.  Lastly, it takes a step in the direction of $\bm{v}^\ell = LO(\tilde{\nabla}^\ell)$ by setting $\bm{x}^{\ell} = \bm{x}^{\ell - 1} + \frac{1}{K} \bm{v}^\ell$ (Lines 10-11). Since $\bm{x}^\ell$ at each iteration is a combination of vertices of $\polym$, the output is guaranteed to be feasible. The intuition for why the algorithm succeeds is that it only moves along nonnegative directions (since $\bm{v}^\ell$ is always nonnegative). This is in contrast to gradient-based algorithms for \emph{concave} optimization, which move in the (possibly negative) direction $\bm{v}^{\ell} - \bm{x}^\ell$. As an up-concave function, $G_\mu$ is concave along all nonegative directions. By moving only in such directions we inherit enough of the nice properties of concave optimization to obtain a $(1 - 1/e)-$ approximation.  

A small technical detail is that adding random noise $\bm{z}$ could result in negative values, for which the multilinear extension is not defined. To circumvent this, we start the algorithm at $\bm{x}^0 = u\bm{1}$ (i.e., with small positive values in every entry) and then return $\bm{x}_{final} = \bm{x}^K - u \bm{1}$ (Line 13). 

\subsection{Theoretical bounds}
Let $T_1$ be the runtime of the linear optimization oracle and $T_2$ be the runtime of the first-order oracle. We prove the following guarantee for SFW:

\begin{theorem}\label{theorem:sfw}
	For any $\epsilon, \delta > 0$, there are parameter settings such that SFW finds a solution $\bm{x}^K$ satisfying $G(\bm{x}^K) \geq (1 - \frac{1}{e})OPT - \epsilon$ with probability at least $1 - \delta$. Its runtime is  $\mathcal{\tilde{O}}\left(T_1 \frac{M^2 k^2 n}{\epsilon^2} + T_2 \frac{k^4 M^4 n}{\epsilon^4} \log\frac{1}{\delta} \right)$\footnote{The $\mathcal{\tilde{O}}$ notation hides logarithmic terms}.
\end{theorem}	

We remark that $T_1$ is small since linear optimization over $\polym$ can be carried out by a greedy algorithm. For instance, the runtime is $T_1 = \mathcal{O}\left(n \log n\right)$ for the uniform matroid, which covers many applications. $T_2$ is typically dominated by the runtime of the BRI since it is known how to efficiently compute the gradient of a submodular function \cite{calinescu2011maximizing,iyer2014monotone}.

Based on this result, we show the following guarantee on a single randomly sampled set that EQUATOR returns after applying swap rounding to the marginal vector $\bm{x}_{final}$.

\begin{theorem}\label{theorem:main}
	With $r=1$, EQUATOR outputs a set $S \in \mathcal{I}$ such that $\min_i \E[f_i(S)] \geq (1 - \frac{1}{e})^2 OPT - \epsilon$ with probability at least $1 - \delta$. Its time complexity is the same as SFW.  
\end{theorem}	

\begin{proof}
	Suppose that $p_{OPT}$ is the distribution achieving the optimal value for Problem $\ref{problem:randomized}$. Let $\bm{x}^*$ be the optimizer for the problem $\max_{\bm{x} \in \polym} G(\bm{x})$. That is, $\bm{x}^*$ can be interpreted as the marginals of the independent distribution which maximizes $\min_i \E_{S \sim p_{\bm{x}^*}^I}[f_i(S)]$. With slight abuse of notation, let $p_{OPT}^I$ be the independent distribution with the same marginals as $p_{OPT}$. By applying the correlation gap to each $f_i \in \mathcal{F}$ and taking the $\min$, we have
	
	\begin{align*}
	\min_{f_i \in \mathcal{F}} \E_{S \sim p_{OPT}}[f_i(S)] &\leq \frac{e}{e - 1} \min_{f_i \in \mathcal{F}} \E_{S \sim p_{OPT}^I}[f_i(S)].
	\end{align*}
	
	By definition of $\bm{x}^*$, $G(\bm{x^*}) \geq \min_{f_i \in \mathcal{F}} \E_{S \sim p_{OPT}^I}[f_i(S)]$. Hence, $G(\bm{x^*}) \geq (1 - 1/e)\min_i \E_{S \sim p_{\bm{x}^*}^I}[f_i(S)] = (1 - 1/e)OPT$. Via Theorem \ref{theorem:sfw}, the marginal vector $\bm{x}$ that our algorithm finds satisfies $G(\bm{x}) \geq (1 - \frac{1}{e}) G(\bm{x}^*) - \epsilon \geq (1 - \frac{1}{e})^2 OPT - \epsilon$. Lastly, Chekuri et al.\ \shortcite{chekuri2010dependent} show that swap rounding outputs an independent set $S$ satisfying $\E[f_i(S)] \geq F_i(S)$ for any $f_i \in \mathcal{F}$, which completes the proof.
\end{proof}

This guarantee is sufficient if we just want to implement the maximizing player's strategy by sampling an action.  We also prove that if a full description of the maximizing player's mixed strategy is desired, drawing a small number of independent samples via swap rounding suffices:

\begin{algorithm}
	\caption{Efficient bicriteria approximation} \label{alg:bicriteria}
	\begin{algorithmic}[1] 
		\State Run EQUATOR to obtain $\bm{x}_{final}$. 
		\For{$j = 1...e\log \frac{1}{\delta}$} 
		\State run SwapRound($\bm{x}_{final}$) $\frac{8\log|\mathcal{F}|}{\epsilon^3}+1$ times, yielding $S^j_1...S^j_r$.
		\Let{$S_j$}{$S_1^j \cup S_2^j \cup....\cup S_r^j$}
		\EndFor
		\State\Return $\arg\max_{S_j} \min_{f_i \in \mathcal{F}} f_i(S_j)$
		
	\end{algorithmic}
\end{algorithm}

\begin{theorem} \label{theorem:sample}
	Draw $r = \mathcal{O}\left(\frac{1}{\epsilon^3} \left(\log |\mathcal{F}| + \log \frac{1}{\delta}\right)\right)$ samples using independent runs of randomized swap rounding. The uniform distribution on these samples is a $(1 - \frac{1}{e})^2 - \epsilon$ approximate equilibrium strategy for the maximizing player with probability at least $1 - \delta$. The runtime is $\mathcal{O}\left(\frac{rk^2M^2n}{\epsilon}\right)$.
\end{theorem}

This also gives a simple way of obtaining a single feasible set (pure strategy) which has a bicriteria guarantee for the robust optimization problem. As pointed out by Chen et al. \shortcite{chen2017robust}, since the $f_i$ are all monotone, taking the union of the sets output by swap rounding gives a single set with at least as much value. Algorithm \ref{alg:bicriteria} implements this procedure. It first solves the fractional problem by running EQUATOR. Then, it carries out a series of independent iterations. Each iteration $j$ draws $\frac{8 \log |\mathcal{F}|}{\epsilon^3}$ sets via swap rounding and stores their union $S_j$. It then returns the best of the $S_j$. Via our concentration bound for the distribution produced in each iteration (Theorem \ref{theorem:sample}), each iteration succeeds in producing a ``good" set with probability at least $\frac{1}{e}$. Algorithm \ref{alg:bicriteria} runs $e \log \frac{1}{\delta}$ iterations so that at least one succeeds with probability at least $1 - \delta$. 

\begin{theorem} \label{theorem:bicriteria}
	Algorithm \ref{alg:bicriteria} returns a single set $S$ which is the union of at most $\frac{8\log|\mathcal{F}|}{\epsilon^3}+1$ elements of $\mathcal{I}$ and satisfies $\min_{f_i \in \mathcal{F}} f_i(S) \geq \left(1 - \frac{1}{e}\right)^2\max_{S^* \in \mathcal{I}} \min_{f_i \in \mathcal{F}} f_i(S^*) - \epsilon$ with probability at least $1 - \delta$. 
\end{theorem}

The strongest existing bicriteria guarantee is for the SATURATE algorithm of Krause et al.\ \shortcite{krause2008robust}, which outputs a set matching the optimal value with size $\left(\log \left(\max_{v \in X}\sum_{f_i \in \mathcal{F}} f_i(\{v\})\right) + 1\right)k$. Our $S$ maintains logarithmic dependence on $|\mathcal{F}|$, but also contains dependence on $\epsilon$. Moreoever, it is only a $(1 - \frac{1}{e})^2$-approximation to the optimal solution quality. However, our result is much more general than that of Krause et al.\ and handles situations that SATURATE cannot. First, our result applies when $\mathcal{F}$ is accessible only through an oracle, where SATURATE relies on explicitly enumerating the functions. Second, our result applies when $\mathcal{I}$ is \emph{any} matroid, where SATURATE applies only to cardinality-constrained problems. To our knowledge, this is the first computationally efficient bicriteria algorithm under either condition. 

\section{Improving the approximation ratio}

In this section, we examine the conditions under which it is possible to improve EQUATOR's $\left(1 - \frac{1}{e}\right)^2$-approximation to $\left(1 - \frac{1}{e}\right)$. The earlier analysis lost a factor $\left(1 - \frac{1}{e}\right)$ in two places: the use of the correlation gap to bound the loss introduced by only tracking marginals, and the use of SFW to solve the continuous relaxation. While the second factor is difficult to improve, we can eliminate the loss from the correlation gap when a stronger best response oracle for the adversary is available. Specifically, we define a \emph{best response to mixture of independent distributions} (BRMI) oracle to be an algorithm which, given a list of marginal vectors $\bm{x}^1...\bm{x}^\rho$, outputs

\begin{align*}
\arg\min_{f_i \in \mathcal{F}} \frac{1}{\rho}\sum_{j = 1}^\rho F_i(\bm{x}^j).
\end{align*}
\begin{algorithm}
	\caption{EQUATOR with improved approximation guarantee} \label{alg:equator-improved}
	\begin{algorithmic}[1] 
		\State Set $\rho = O\left(\frac{W^2 \log |\mathcal{F}|}{\epsilon^2}\right)$
		\State Use SFW to solve the problem $\max_{\bm{x}^1..\bm{x}^\rho \in \times_{j = 1}^\rho\mathcal{P}} \min_{f_i \in \mathcal{F}} \frac{1}{\rho}\sum_{j = 1}^\rho F_i(\bm{x}^j)$, obtaining $\bm{x}^1..\bm{x}^\rho$
		\State Set $r = \mathcal{O}\left(\frac{1}{\epsilon^3}\log \left( \frac{|\mathcal{F}|\rho}{\delta}\right)\right)$
		\For{$i = 1...\rho$}
		\State Draw sets $S_1^i...S_r^i$ independently as SwapRound($\bm{x}^i$). 
		\EndFor
		\State Return the uniform distribution on $\{S_j^i : i = 1...\rho, j = 1...r\}$. 
	\end{algorithmic}
\end{algorithm}
We will be interested in BRMI oracles which take time polynomial in $\rho$. As the name implies, a BRMI oracle can compute adversary best responses to any distribution which is explicitly represented as a mixture of independent distributions with given marginals. By contrast, a BRI is restricted to a single independent distribution. A BRMI is a considerably more powerful oracle because, with sufficiently large $\rho$, any distribution can be arbitrarily well-approximated by a mixture of independent distributions (a statement which is formalized below). Hence, the algorithm we propose maintains $\rho$ copies of the decision variables $\bm{x}^1...\bm{x}^\rho$ for a value of $\rho$ which will be set later. We aim to maximize

\begin{align*}
\max_{\bm{x}^1..\bm{x}^\rho \in \times_{j = 1}^\rho\mathcal{P}} \min_{f_i \in \mathcal{F}} \frac{1}{\rho}\sum_{j = 1}^\rho \E_{S \sim \bm{x}^j}[f_i(S)]
\end{align*}

which we recognize as being equivalent to the problem

\begin{align}\label{problem:continuous-mixture}
\max_{\bm{x}^1..\bm{x}^\rho \in \times_{j = 1}^\rho\mathcal{P}} \min_{f_i \in \mathcal{F}} \frac{1}{\rho}\sum_{j = 1}^\rho F_i(\bm{x}^j)
\end{align}

It is easy to check that $\frac{1}{\rho}\sum_{j = 1}^\rho F_i(\bm{x}^j)$ is an up-concave function which inherits all of the smoothness properties of the $F_i$. Hence, we can use SFW to obtain a $\left(1 - \frac{1}{e}\right)$-approximate solution to Problem \ref{problem:continuous-mixture} provided that we have a BRMI oracle with which to compute gradients. After solving Problem \ref{problem:continuous-mixture}, we can use swap rounding to produce feasible sets with guaranteed approximation ratio. For a single set, we first select a $j \in \{1...\rho\}$ uniformly at random and then run swap rounding on $\bm{x}^j$. To output a full distribution, as in Theorem \ref{theorem:sample}, we draw $r = \mathcal{O}\left(\frac{1}{\epsilon^3}\log\left(\frac{|\mathcal{F}| \rho }{\delta}\right)\right)$ samples from each of the $\bm{x}^j$ and then output the uniform distribution over the combined set of samples. The extra logarithmic dependence on $\rho$ ensures that we can take a final union bound over the $\rho$ batches of swap rounding. The entire procedure is summarized in Algorithm \ref{alg:equator-improved}. We let $W$ be an upper bound on the value of $f$ for any feasible set: $W \geq \max_{f_i \in \mathcal{F}, S \in \mathcal{I}} f_i(S)$. Note that $W \leq nM$ always holds via submodularity, but tighter bounds might apply for particular functions. 

We have the following approximation guarantee for Algorithm \ref{alg:equator-improved}. We note that the idea of optimizing over a mixture of independent distributions has been used in \cite{dughmi2017algorithmic}, but we prove Lemma \ref{lemma:mixture-sample} (establishing that a good mixture exists) for completeness. 

\begin{theorem}
	Given access to a BRMI oracle for any SBR game instance, Algorithm \ref{alg:equator-improved} returns a distribution $p$ which satisfies $\min_{f_i \in \mathcal{F}}\E_{S \sim p}\left[f_i(S)\right] \geq \left(1 - \frac{1}{\epsilon}\right)OPT - \epsilon$ with probability at least $1 - \delta$. 
\end{theorem}

\begin{proof}
	We first establish that there exists a near-optimal distribution over elements of $\mathcal{I}$ with support size at most $O\left(\frac{W^2 \log |\mathcal{F}|}{\epsilon^2}\right)$:
	
	\begin{lemma}\label{lemma:mixture-sample}
		Take any collection of functions $\mathcal{F}$ with $\max_{f_i \in \mathcal{F}, S \in \mathcal{I}} f_i(S) \leq W$ and a distribution $p \in \Delta(\mathcal{I})$. There exists a distribution $q$ supported on at most $\rho = O\left(\frac{W^2 \log |\mathcal{F}|}{\epsilon^2}\right)$ elements of $\mathcal{I}$ which satisfies $\E_{S \sim q} [f_i(S)] \geq \E_{S \sim p}[f_i(S)] - \epsilon$ for all $f_i \in \mathcal{F}$. 
	\end{lemma} 
	\begin{proof}
		We will use the probabilistic method. Suppose that we draw $\rho = \frac{W^2 \log |\mathcal{F}|}{\epsilon^2}$ samples $S_1...S_\rho$ independently from $p$ and let $q$ be the uniform distribution on the samples. Fix an arbitrary function $f_i$. Via Hoeffding's inequality, we have that
		
		\begin{align*}
		\text{Pr}\left[\E_{S \sim p}\left[f_i(S)\right] - \frac{1}{r}\sum_{i = 1}^r f_i(S) \geq \epsilon \right] \leq e^{-\frac{2r \epsilon^2}{W^2}}
		\leq \frac{1}{|\mathcal{F}|^2}
		\end{align*}
		
		and this holds simultaneously for all scenarios $y$ with probability at least $1 - \frac{1}{|\mathcal{F}|} > 0$ via union bound. That is, we have a random sampling procedure which outputs a distribution $q$ satisfying $\E_{S \sim q}[f_i(S)] \geq \E_{S \sim p}[f_i(S)] - \epsilon$ for all $f_i \in \mathcal{F}$ with positive probability. Via the probabilistic method we are guaranteed that such a distribution (i.e., one which is a uniform distribution on at most $\rho$ elements of $\mathcal{I}$) exists. 
	\end{proof}

	Now, note that Algorithm \ref{alg:equator-improved} maximizes over the set $\times_{j = 1}^\rho\mathcal{P}$, which includes the distribution $q$. Via the guarantee for SFW (Theorem \ref{theorem:sfw}), SFW returns $\bm{x}^1...\bm{x}^\rho$ satisfying $ \min_{f_i \in \mathcal{F}} \frac{1}{\rho}\sum_{j = 1}^\rho F_i(\bm{x}^j) \geq \left(1 - \frac{1}{e}\right)OPT - \epsilon$ (we ignore for convenience the issue of adjusting all of the $\epsilon$ values by a constant factor). Now we just need to establish that the rounding procedure succeeds. A simple variation on the proof of Theorem \ref{theorem:sample} suffices: we claim that $\frac{1}{r}\sum_{a = 1}^r f_i(S_j^a) \geq \E_{S \sim \bm{x}^j}[f_i(S)] - \epsilon$ holds for each $i, j$ with probability at least $1 - \frac{\delta}{\rho |\mathcal{F}|}$ via our choice of $r$. Taking union bound over all $i = 1...|\mathcal{F}|$ and $j = 1...r$ completes the proof. 
\end{proof}

\section{Applications}

We now give several examples of domains that our algorithm can be applied to. In each of these cases, we obtain the first guaranteed polynomial time constant-factor approximation algorithm for the problem. The key part of both applications is developing a BRI (the first order oracle is easily obtained in closed form via straightforward calculus). 

\textbf{Network security games: } Earlier, we formulated network security games in the SBR framework. All we need to solve it using EQUATOR is a BRI oracle. The full attacker best response problem is known to be NP-hard \cite{jain2011double}. However, it turns out the best response to an \emph{independent} distribution is easily computed. Index the set of paths and let $P_i$ be the $i$th path, ending at a target with value $\tau_i$. Let $P(t_j)$ be the set of all paths from the (super)source $s$ to $t_j$. Let $f_i$ be the corresponding submodular objective. Given a defender mixed strategy $\bm{x}$, the attacker best response problem is to find $\min_i \E_{S \sim \bm{x}}[f_i(S)]$. We can rewrite this as

\begin{align*}
\min_i \E_{S \sim \bm{x}}[f_i(S)] &= \min_i \E_{S \sim \bm{x}}[\tau_i \bm{1}[S \cap P_i \not= \emptyset]]\\
&= \min_{t_j \in T} \tau_j \min_{P \in P(t_j)} \E_{S \sim \bm{x}}[\bm{1}[S \cap P \not= \emptyset]]\\
&= \min_{t_j \in T} \tau_j \min_{P \in P(t_j)} 1 - \prod_{e \in P} \left[1 - x_e\right]
\end{align*}

We can now solve a separate problem for each target $t_j$ and then take the one with lowest value. For each $t_j$, we solve a shortest path problem. We aim to find a $s-t_j$ path which maximizes the product of the the weights $1 - x_e$ on each edge. Taking logarithms, this is equivalent to finding the path which minimizes $-\sum_{e \in P}\log (1 - x_e) = \sum_{e \in P}\log \frac{1}{1 - x_e}$. This is a shortest path problem in which each edge has nonnegative weight $\log \frac{1}{1 - x_e}$, and so can be solved via Dijkstra's algorithm. With the attacker BRI in hand, applying EQUATOR yields the first subexponential-time algorithm for network security games.

\textbf{Robust coverage and budget allocation: }Many widespread applications of submodular functions concern coverage functions. A coverage function takes the following form. There a set of items $U$, and each $j \in U$ has a weight $w_j$. The algorithm can choose from a ground set $X = \{a_1...a_n\}$ of actions. Each action $a_i$ covers a set $A_i \subseteq U$. The value of any set of actions is the total value of the items that those actions cover: $f(S) = \sum_{j \in \bigcup_{i \in S} A_i} w_j$. We can also consider probabilistic extensions where action $a_i$ covers each $j \in A_i$ independently with probability $p_{ij}$. This framework includes budget allocation, sensor placement, facility location, and many other common submodular optimization problems. Here we consider a \emph{robust coverage} problem where the weights $\bm{w}$ are unknown. For concreteness, we focus on the budget allocation problem, but all of our logic applies to general coverage functions. 


Budget allocation models an advertiser's choice of how to divide a finite budget $B$ between a set of advertising channels. Each channel is a vertex on the left hand side $L$ of a bipartite graph. The right hand $R$ consists of customers. Each customer $v \in R$ has a value $w_v$ which is the advertiser's expected profit from reaching $v$. The advertiser allocates their budget in integer amounts among $L$. Let $y(s)$ denote the amount of budget allocated to channel $s \in L$. The advertiser solves the problem

\begin{align*}
\max_{\bm{y}: ||\bm{y}||_1 \leq B} f_{\bm{w}}(\bm{y}) = \sum_{v \in R} w_v \left[1 - \prod_{s \in L} (1 - p_{sv})^{y(s)}\right] 
\end{align*}

where $p_{sv}$ is the probability that one unit of advertising on channel $s$ will reach customer $v$. This a probabilistic coverage problem where the action set $X$ contains $B$ copies\footnote{We use this formulation for simplicity, but it is possible to use only $\log B$ copies of each node \cite{ene2016reduction}.} of each $s \in L$ and the feasible decisions $\mathcal{I}$ are all size $B$ subsets of $X$. Choosing $b$ copies of node $s$ corresponds to setting $y(s) = b$. Budget allocation has been the subject of a great deal of recent research \cite{alon2012optimizing,soma2014optimal,miyauchi2015threshold}.

In the robust optimization problem, the profits $\bm{w}$ are not exactly known. Instead, they belong to a polyhedral uncertainty set $\mathcal{U}$. This is very realistic: while an advertiser may be able to estimate the profit for each customer from past data, they are unlikely to know the true value for any particular campaign. We remark that Staib and Jegelka \shortcite{staib2017robust} also considered a robust budget allocation problem, but their problem has uncertainty on the probabilities $p_{st}$, not the profits $\bm{w}$. Further, they consider a continuous problem without the complication of rounding to discrete solutions. 

As an example uncertainty set, consider the D-norm uncertain set, which is common in robust optimization \cite{bertsimas2004robust,staib2017robust}. The uncertainty set is defined around a point estimate $\bm{\hat{w}}$ as 

\begin{align*}
\mathcal{U}_\gamma^{\bm{\hat{w}}} = \{\bm{w} : \exists \bm{c} \in [0,1]^{|R|}, w_i = (1 - c_i)\hat{w}_i, \,\, ||\bm{c}||_1 \leq \gamma\}.
\end{align*}

This can be thought of as allowing an adversary to scale down each entry of $\bm{\hat{w}}$ with a total budget of $\gamma$. In our case, $\bm{\hat{w}}$ is the advertiser's best estimate from past data, and they would like to perform well for all scenarios within $\mathcal{U}_\gamma^{\bm{\hat{w}}}$. $\gamma$ defines the advertiser's tolerance for risk. The problem we want to solve is $\max_{p \in \Delta(\mathcal{I})} \min_{\bm{w} \in \mathcal{U}_\gamma^{\hat{\bm{w}}}} \E_{\bm{y} \sim p}[f_{\bm{w}}(\bm{y})]$, which we recognize as an instance of Problem \ref{problem:randomized}. For any fixed distribution $p$, we have by linearity of expectation 
\begin{align*}
\E_{\bm{y} \sim p}[f_{\bm{w}}(\bm{y})] = \sum_{v \in R} w_v \E_{\bm{y} \sim p}\left[1 - \prod_{s \in L} (1 - p_{sv})^{y(s)}\right].
\end{align*}

Note that the inner expectation (which is the total probability that each $v\in R$ is reached) is constant with respect to $\bm{w}$. Hence, the adversary's best response problem of computing $\min_{\bm{w} \in \mathcal{U}} \E_{\bm{y} \sim p}[f_{\bm{w}}(\bm{y})]$ is a linear program and can be easily solved. The coefficients of this LP (the inner expectation in the above sum) can easily be computed exactly for any independent distribution. Further, since any LP has an optimal solution among the vertices of $\mathcal{U}_\gamma^{\hat{\bm{w}}}$, we can without loss of generality restrict the adversary's pure strategies to a finite (though exponentially large) number. 

Lastly, we remark that it also possible to obtain a BRMI for this problem. For any distribution $p$, we can find a best response via linear programming provided that the coefficients $\E_{\bm{y} \sim p}\left[1 - \prod_{s \in L} (1 - p_{sv})^{y(s)}\right]$ can be computed for each $v \in R$. This is easy when $p$ is given explicitly as a mixture of independent distributions $\bm{x}^1...\bm{x}^\rho$ since we just average over the corresponding term for each individual $\bm{x}^i$. Hence, we can use Algorithm \ref{alg:equator-improved} to obtain a $\e$-approximation. Nevertheless, we use the original EQUATOR algorithm in our experiments and find that it performs near-optimally despite its theoretically weaker approximation ratio.
%
\begin{figure*}
	\centering
	\includegraphics[height=1.1in]{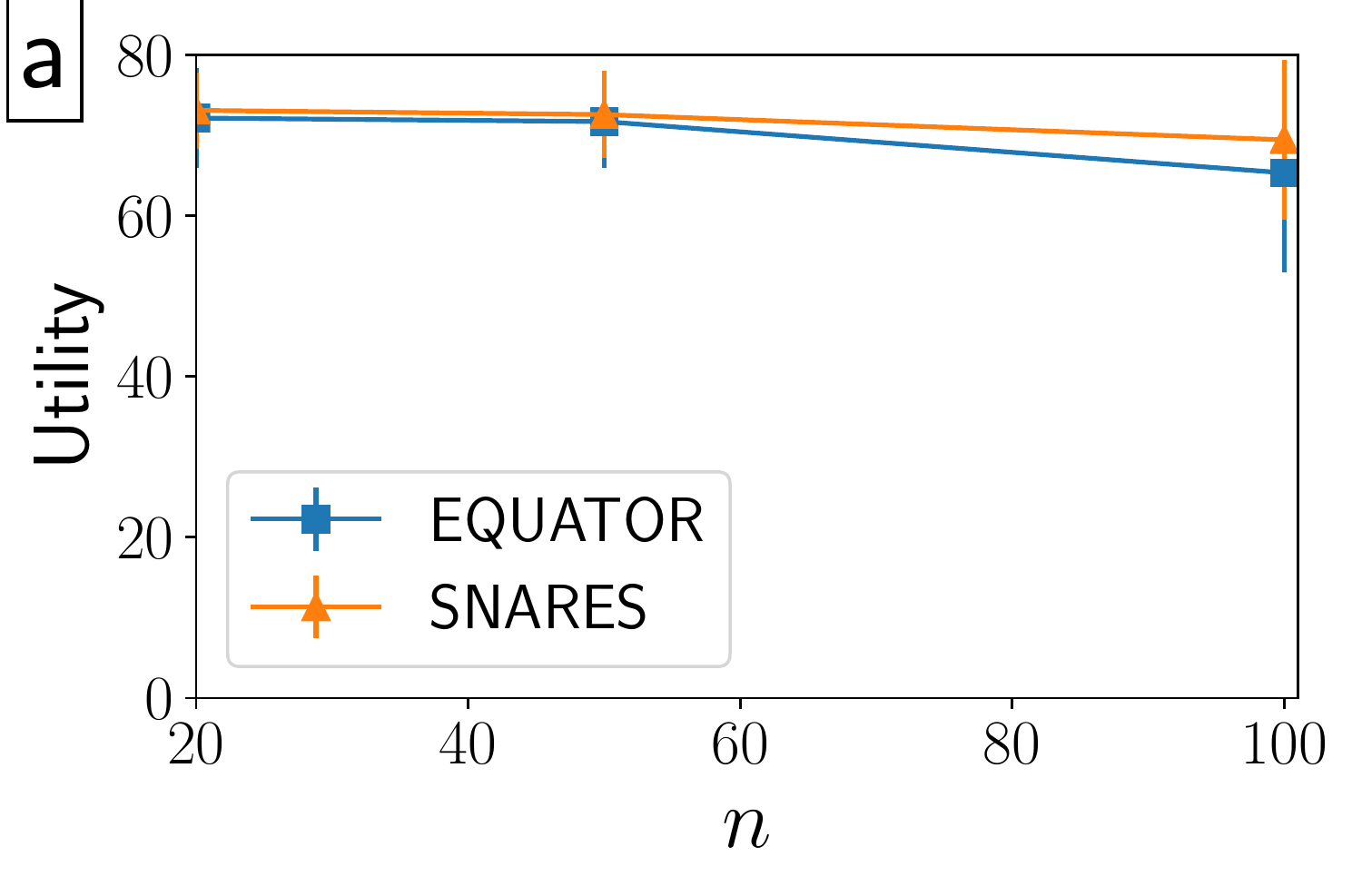}
	\includegraphics[height=1.1in]{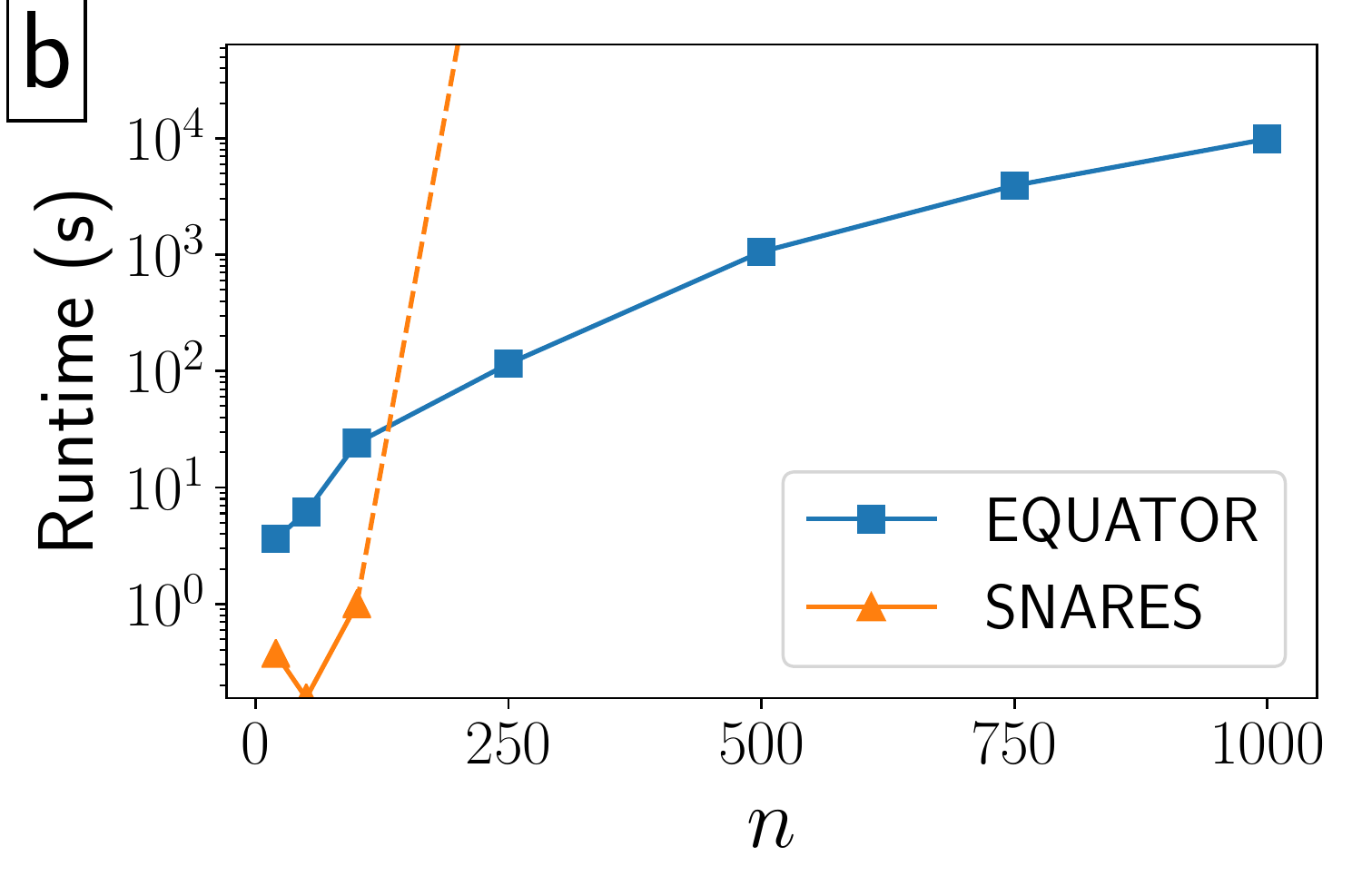}
	
	\includegraphics[height=1.1in]{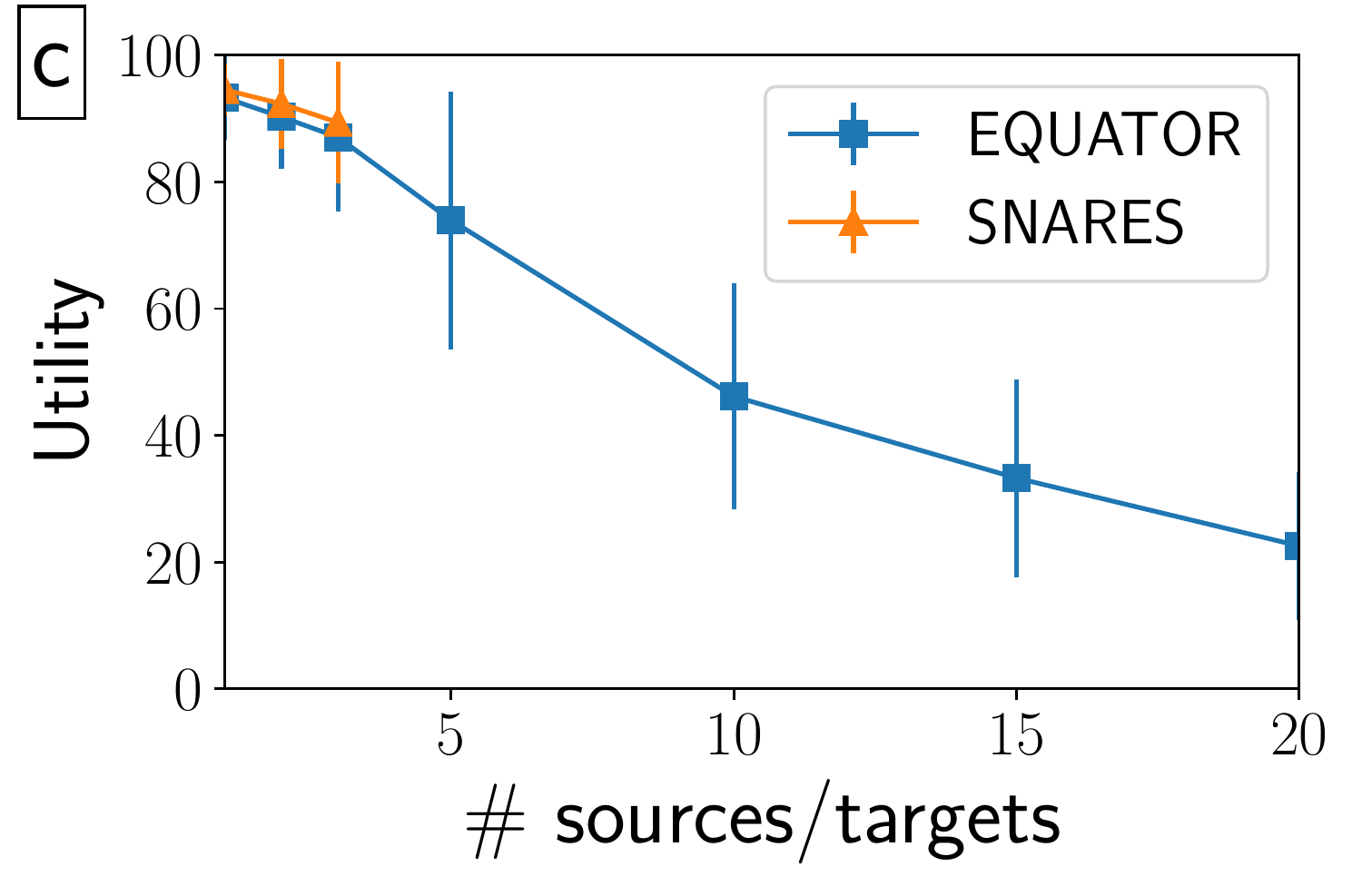}
	\includegraphics[height=1.1in]{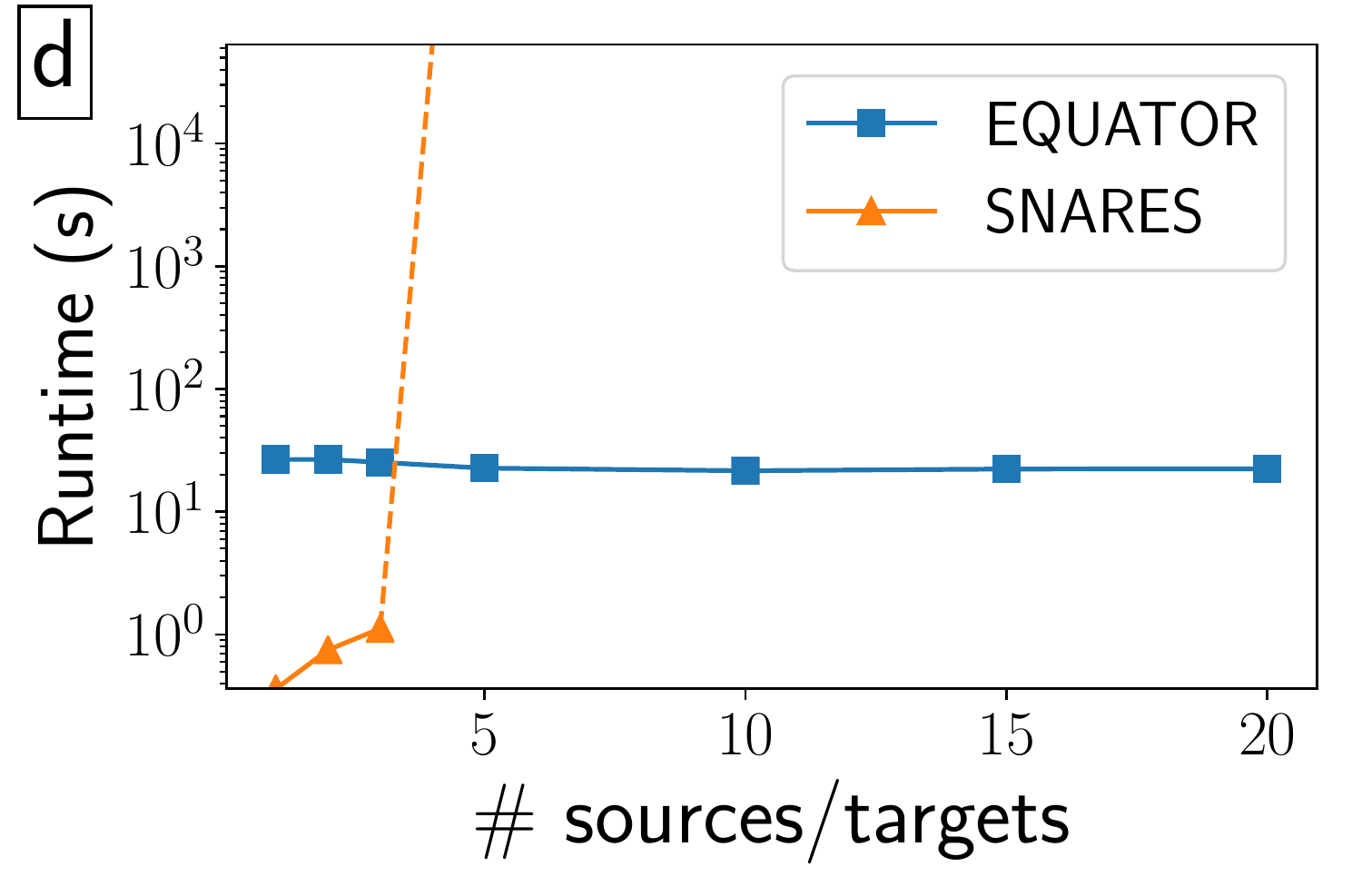}
	\caption{Experimental results for network security games.}\label{fig:nsg}
\end{figure*}

\section{Experiments}

We now show experimental results from applying EQUATOR to these two domains.

\textbf{Network security games:} We first study the network security game defined above. We compare EQUATOR to the SNARES algorithm \cite{jain2013security} which is the current state of the art algorithm with guaranteed solution quality. SNARES uses a double oracle approach to find a \emph{globally} optimal solution. However, it incorporates several domain-specific heuristics which substantially improve its runtime over a standard implementation of double oracle. We note that Iwashita et al.\ \shortcite{iwashita2016simplifying} proposed a newer double-oracle style algorithm which first preprocesses the graph to remove unnecessary edges. We do not compare to this approach because the preprocessing step can be applied equally well to either EQUATOR or double oracle.  We use random geometric graphs, which are commonly used to assess algorithms for this domain due to their similarity to real world road networks \cite{jain2013security,iwashita2016simplifying}. As in Jain et al.\ \shortcite{jain2013security}, we use density $d = 0.1$ with the value of each target drawn uniformly at random in $[0, 100]$. We set $k$ to be one percent of the number of edges. Each data point averages over 30 random instances. EQUATOR was run with $K = 100, c = 60, u = 0.1$.

Figure \ref{fig:nsg} shows the results. Figures \ref{fig:nsg}(a) and \ref{fig:nsg}(b) vary the network size $n$ with three randomly chosen source and target nodes. Figure \ref{fig:nsg}(a) plots utility (i.e., how much loss is averted by the defender's allocation) as a function of $n$. Error bars show one standard deviation. We see that EQUATOR obtains utility within 6\% of SNARES, which computes a global optimum. Figure \ref{fig:nsg}(b) shows runtime (on a logarithmic scale) as a function of $n$. SNARES was terminated after 10 hours for graphs with 250 nodes, while EQUATOR easily scales to 1000 nodes. Next, Figures \ref{fig:nsg}(c) and \ref{fig:nsg}(d) show results as the number of sources and targets grows. As expected, utility decreases with more sources/targets since the number of resources is constant and it becomes harder to defend the network. EQUATOR obtains utility within 4\% of SNARES. However, SNARES was terminated after 10 hours for just 5 source/targets, while EQUATOR runs in under 25 seconds with 20 source/targets. 

\textbf{Robust budget allocation: }We compare three algorithms for robust budget allocation. First, EQUATOR. Second, double oracle. We use the greedy algorithm for the defender's best response (which is a $(1 - 1/e)$-approximation) since the exact best response is intractable. For the adversary's best response, we use the linear program discussed in the section on robust coverage. Third, we compare to ``greedy", which greedily optimizes the advertiser's return under the point estimate $\bm{\hat{w}}$. Greedy was implemented with lazy evaluation \cite{minoux1978accelerated} which greatly improves its runtime at no cost to solution value. We generated random bipartite graphs with $|L| = |R| = n$ where each potential edge is present with probability $0.2$ and for each edge $(u,v)$, $p_{u,v}$ is draw uniformly in $[0, 0.2]$. $\hat{\bm{w}}$ was randomly generated with each coordinate uniform in $[0.5, 1.5]$. Our uncertainty set is the D-norm set around $\bm{\hat{w}}$ with $\gamma = \frac{1}{2}n$, representing a substantial degree of uncertainty. The budget was $B = 5 + 0.01 \cdot n$ since the problem is hardest when $B$ is small relative to $n$. EQUATOR was run with $K = 20, c = 10, u = 0.1$. 

Figure \ref{fig:budget} shows the results. Each point averages over 30 random problem instances (error bars would be hidden under the markers). Figure \ref{fig:budget}(a) plots the profit obtained by each algorithm when the true $\bm{w}$ is chosen as the worst case in $\mathcal{U}_\gamma^{\bm{\hat{w}}}$, with $n$ increasing on the $x$ axis. Figure \ref{fig:budget}(b) plots the average runtime for each $n$. We see that double oracle produces highly robust solutions. However, for even $n = 500$, its execution was halted after 10 hours. Greedy is highly scalable, but produces solutions that are approximately 40\% less robust than double oracle. EQUATOR produces solution quality within 7\% of double oracle and runs in less than 30 seconds with $n = 1000$. 

\begin{figure*}
	\centering
	\includegraphics[height=1.1in]{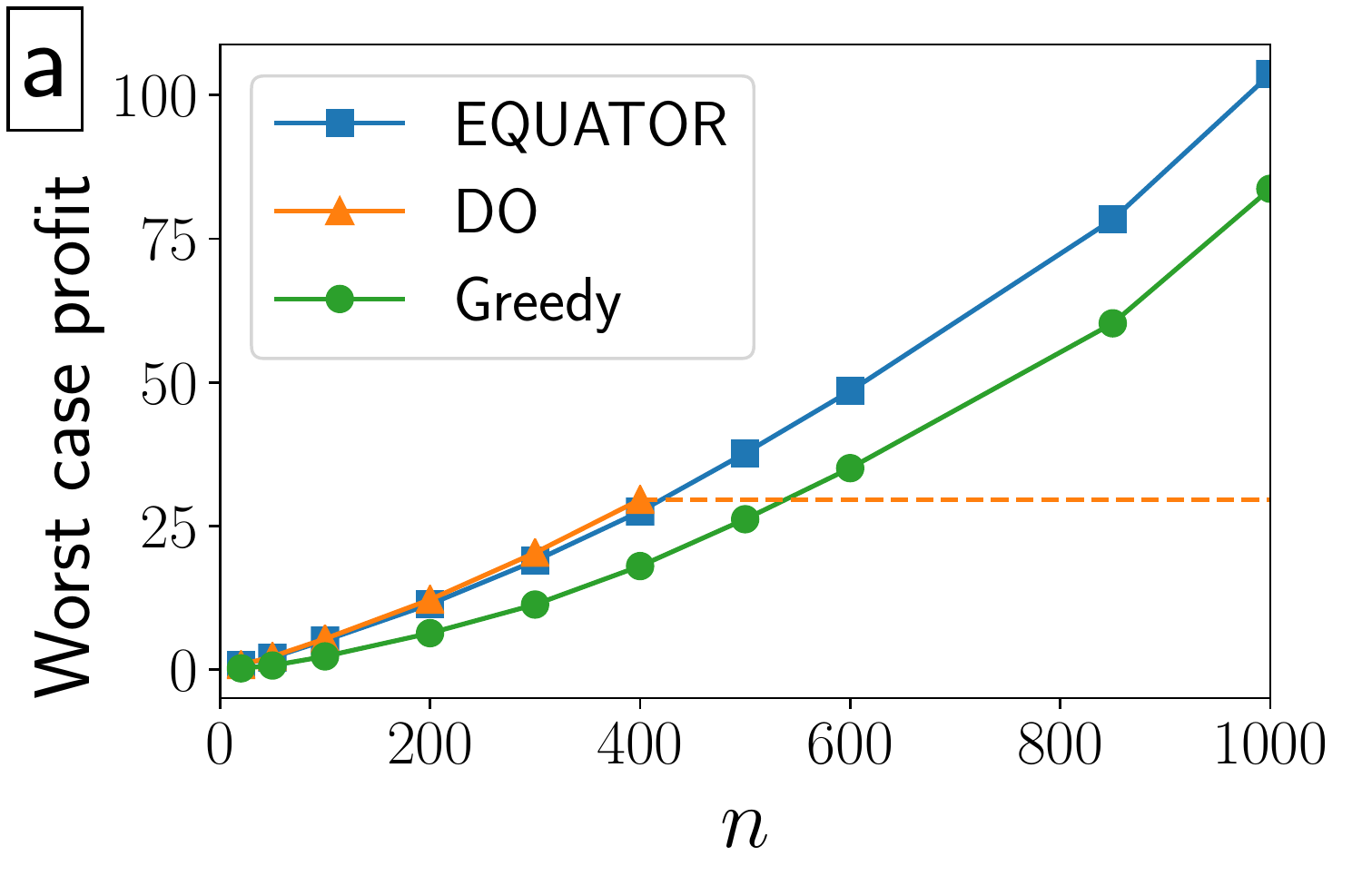}
	\includegraphics[height=1.1in]{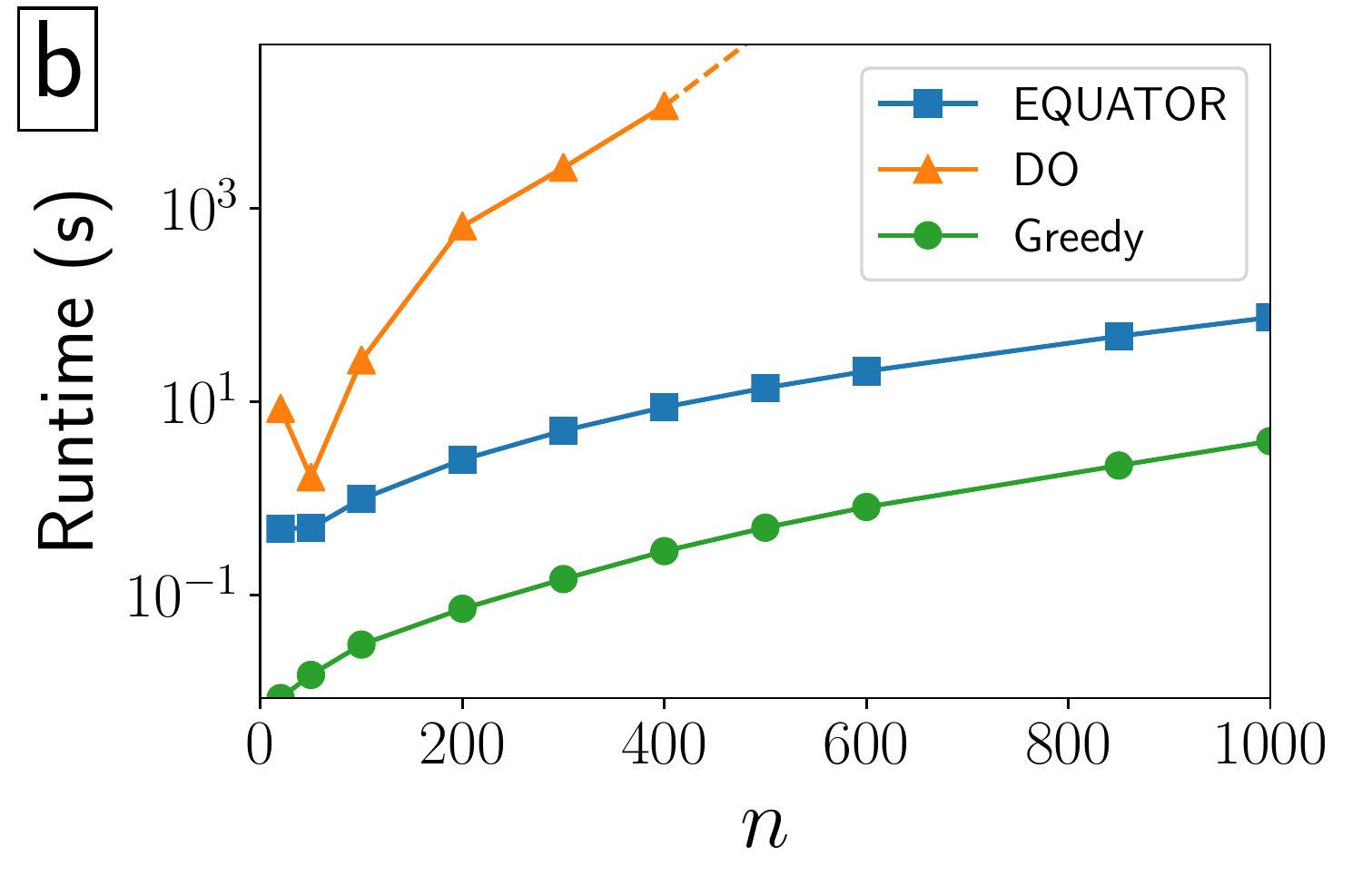}
	
	\includegraphics[height=1.1in]{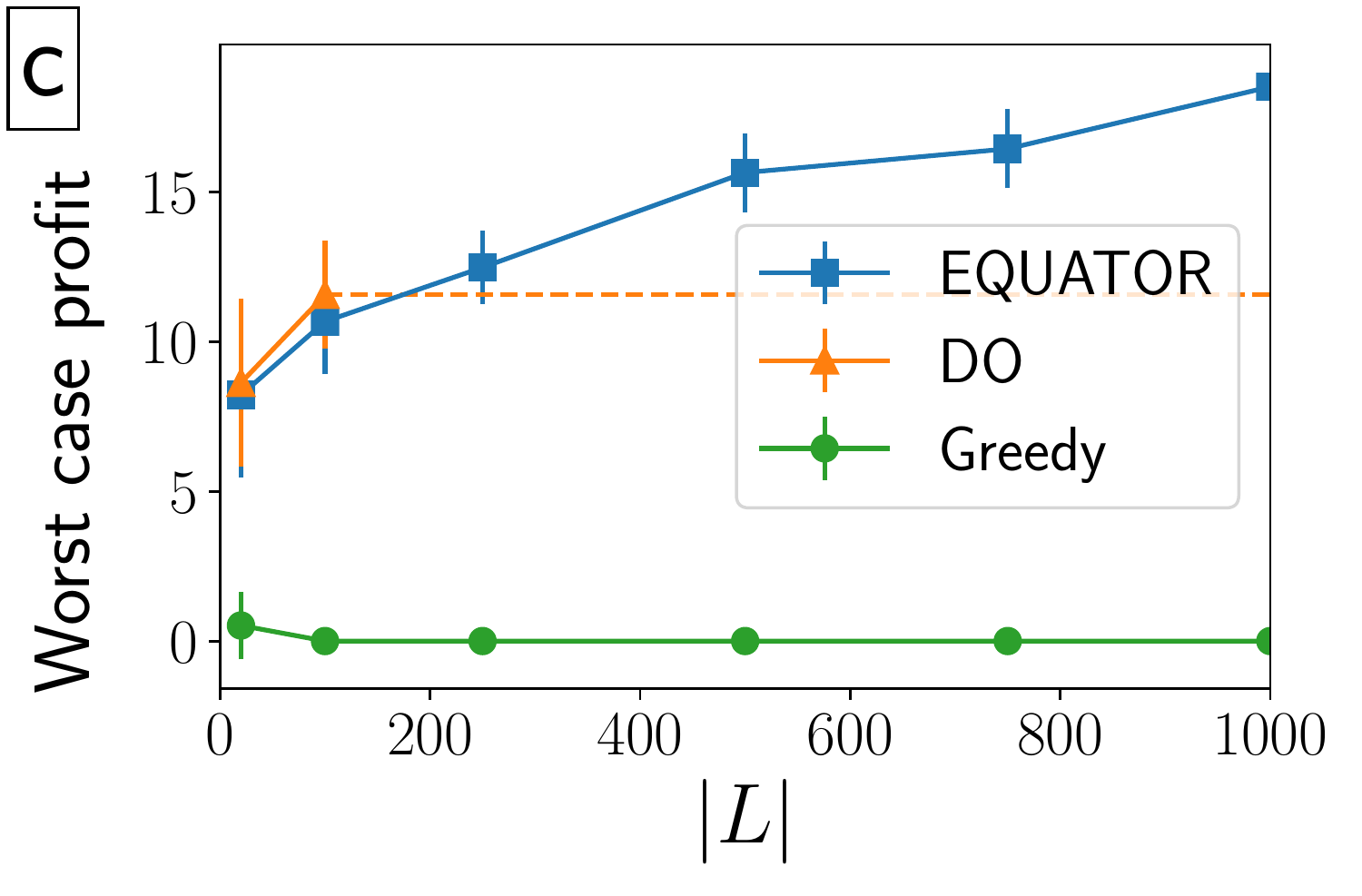}
	\includegraphics[height=1.1in]{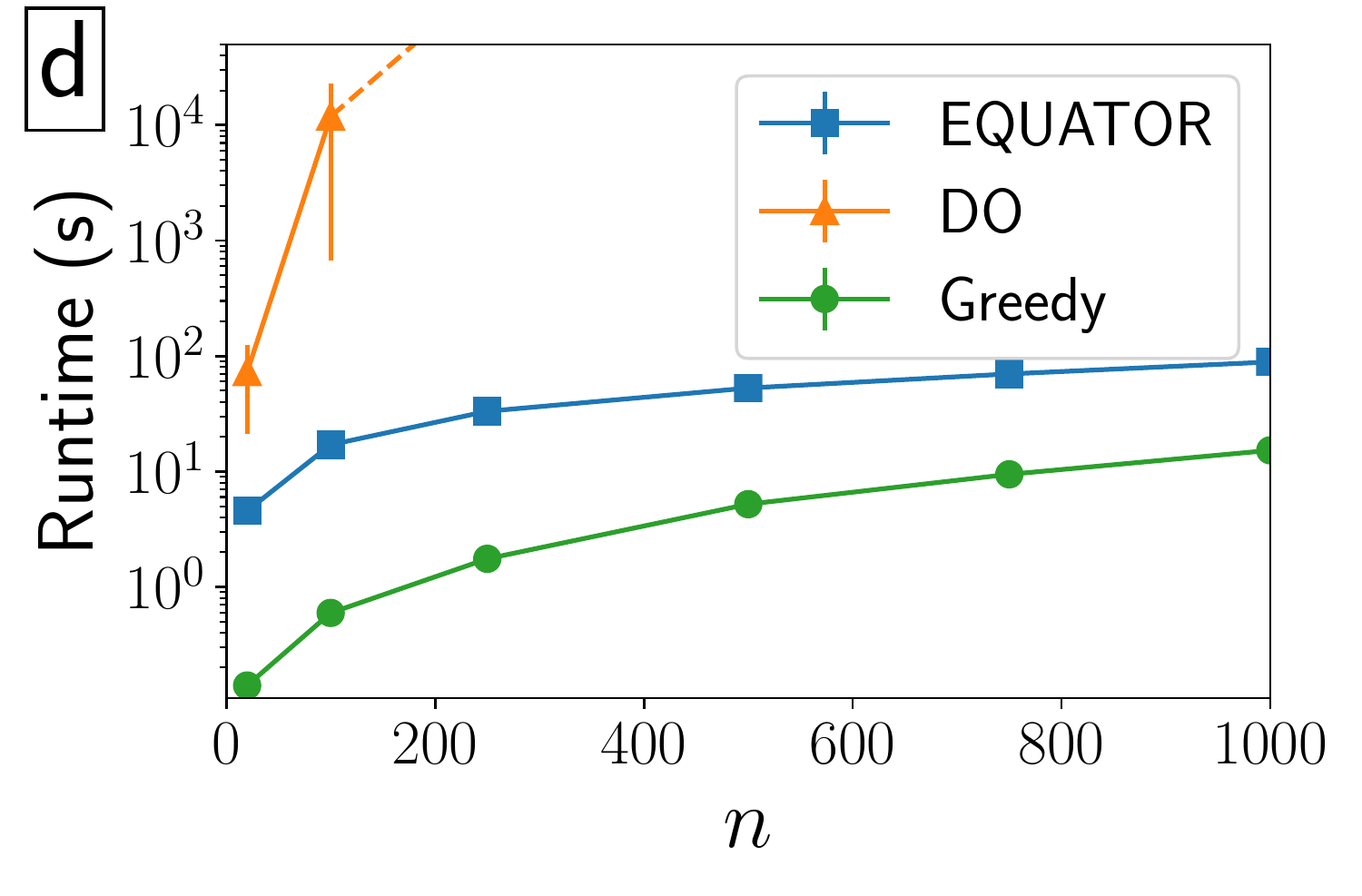}
	\caption{Experimental results for budget allocation.} \label{fig:budget}
\end{figure*}

Next, we show results on a real world dataset from Yahoo webscope \cite{yahoowebscope}. The dataset logs bids placed by advertisers on a set of phrases. We create a budget allocation problem where the phrases are advertising channels and the accounts are targets; the resulting problem has $|L| = 1000$ and $|R| = 10,394$. Other parameters are the same as before. We obtain instances of varying size by randomly sampling a subset of $L$. Figures \ref{fig:budget}(c-d) show results (averaging over 30 random instances). In Figure \ref{fig:budget}(c), we see that both double oracle and EQUATOR find highly robust solutions, with EQUATOR's solution value within 8\% of that of double oracle. By contrast, greedy obtains \emph{no} profit in the worst case for $|L| > 20$, validating the importance of robust solutions on real problems. In Figure \ref{fig:budget}(d), we observe that double oracle was terminated after 10 hours for $n = 500$ while EQUATOR scales to $n = 1000$ in under 40 seconds. Hence, EQUATOR is empirically successful at finding highly robust solutions in an efficient manner, complementing its theoretical guarantees.

\section*{Discussion and conclusion}

This paper introduces the class of submodular best response games, capturing the zero sum interaction between two players when one has a submodular best response problem. Examples include network security games and robust submodular optimization problems. We study the case where the set of possible objective functions is very large (exponential in the problem size), arising from an underlying combinatorial structure. Our main result is a pseudopolynomial time algorithm to compute an approximate minimax equilibrium strategy for the maximizing player when the set of submodular objectives admits a certain form of best response oracle. We instantiate this framework for two example domains, and show experimentally that our algorithm scales to much larger instances than previous approaches. 

One interesting direction for future work is to extend this framework to new application domains. Submodular structure is present in many problems, e.g., sensor placement in water networks \cite{krause2008robust} or cyber-security monitoring \cite{haghtalab2015monitoring}. Both seem natural domains for future work, but designing appropriate best response oracles may be algorithmically challenging. Another open direction is to extend our framework to cases where only \emph{approximate} best responses are available for the adversary. This would enable applications even in settings where an exact BRI is computationally intractable.

\textbf{Acknowledgments}: This research was supported by a NSF Graduate Fellowship. We thank Shaddin Dughmi for helpful conversations.


\bibliographystyle{plainnat}
\bibliography{submodular_games_bib}

\appendix
\section{Appendix: Omitted proofs}

We start out by proving some lemmas from the main text.

\begin{proof}[Proof of Lemma \ref{lemma:min-upconcave}]
	Let $\bm{u} \succeq 0$. We would like to show that for any $\bm{x}$ and any $\xi \geq 0$, $G(\bm{x} + \xi \bm{u})$ is concave as a function of $\xi$. Fix any $\xi_1, \xi_2 \geq 0$ and any $\lambda \in [0,1]$. We have
	
	\begin{align*}
	\min_i F_i(\bm{x} + (\lambda \xi_1 + (1 - \lambda)\xi_2)\bm{u}) &\geq \min_i \left[\lambda F_i(\bm{x} + \xi_1 \bm{u}) + (1 - \lambda)F_i(\bm{x} + \xi_2 \bm{u})\right]\\
	&\geq \lambda \min_i F_i(\bm{x} + \xi_1 \bm{u}) + (1 - \lambda) \min_i F_i(\bm{x} + \xi_2 \bm{u})
	\end{align*}
	
	where the first inequality follows because each $F_i$ is individually up-concave. 
\end{proof}

\begin{proof}[Proof of Lemma \ref{lemma:lipschitz}]
	$G$ is differentiable at a point $\bm{x}$ precisely when there is a unique $F_i$ such that $F_i(\bm{x}) = \min_j F_j(\bm{x})$. Here, we have $\nabla G(\bm{x}) = \nabla F_i(\bm{x})$. Note that $\frac{\partial F_i}{\partial x_j}\Big\rvert_{\bm{x}} = \E[f_i(R(\bm{x})|j \in R(\bm{x}))] - \E[f_i(R(\bm{x})|j \not\in R(\bm{x}))] = \E[f_i(j | R(\bm{x}  - x_j))]$. By submodularity, we conclude that $\frac{\partial F_i}{\partial x_j}\Big\rvert_{\bm{x}} \leq f_i(\{j\}) \leq M$. Further, $\frac{\partial F_i}{\partial x_j}\Big\rvert_{\bm{x}} \geq 0$ always holds by monotonicity. Thus, $||\nabla G(\bm{x})||_\infty \leq M$.
\end{proof}

Let $\mu$ be the uniform probability distribution over the $\ell_\infty$ ball of radius $u$. Define the smoothed function $G_\mu(\bm{x}) = \E_{\bm{z} \sim \mu}[G(\bm{x} + \bm{z})]$. We will show the following properties of $G_\mu$:

\begin{proof}[Proof of Lemma \ref{lemma:smooth-approx}]
	
	For the first property, we start out by fixing the draw of $\bm{z}$ from $\mu$. Following the logic of Lemma \ref{lemma:min-upconcave}, we have that 
	
	\begin{align*}
	\min_i F_i(\bm{x} + \bm{z} + (\lambda \xi_1 + (1 - \lambda)\xi_2)\bm{u}) &\geq \min_i \lambda F_i(\bm{x} + \bm{z} \xi_1 \bm{u}) + (1 - \lambda)F_i(\bm{x} + \bm{z} + \xi_2 \bm{u})\\
	&\geq \lambda \min_i F_i(\bm{x} + \bm{z} + \xi_1 \bm{u}) + (1 - \lambda) \min_i F_i(\bm{x} + \bm{z} + \xi_2 \bm{u}).
	\end{align*}
	
	Since these inequalities hold for any fixed $\bm{z}$, they also hold in expectation over a random $\bm{z}$, so we conclude that $G_\mu$ is up-concave.
	
	For the second property: since $||\nabla G||_\infty \leq M$, $G$ is $M$-Lipschitz with respect to the $\ell_1$ norm. Thus, we have 
	\begin{align*}
	\E[G(\bm{x} + \bm{z})] \leq G(\bm{x}) + M\E[||\bm{z}||_1] \leq G(\bm{x}) + \frac{Mnu}{2}
	\end{align*}
	and analogously, $\E[G(\bm{x} + \bm{z})] \geq G(\bm{x}) - \frac{Mnu}{2}$.
	
	The third property follows from the fact that $G$ is differentiable almost everywhere. To see this, note that $G$ is differentiable wherever there is a unique minimizing $F_i$, in which case $\nabla G = \nabla F_i$. Suppose that there is not a unique minimizer at some point $\bm{x}$. There are two cases. First, if there is an open ball around $\bm{x}$ such that the minimizing functions at $\bm{x}$ coincide at every point in the ball, then their gradients also coincide in the ball. Thus, $G$ is still differentiable at $\bm{x}$. Second, if no such open ball exists, then the set of points at which $G$ is not differentiable has measure zero. Hence, taking a random perturbation of the input avoids such points with probability 1.
	
	For the proof of the fourth property, we follow the argument of Duchi et al.\ (2012). We first claim that 
	
	\begin{align} \label{eq:duchi-claim}
	||\nabla G_\mu(\bm{x}) - \nabla G_\mu(\bm{y})||_\infty = ||\E\left[\nabla G(\bm{x} + \bm{z})\right] - \E\left[\nabla G(\bm{y} + \bm{z})\right]||_\infty  \leq M \int | \mu(\bm{z} - \bm{x}) - \mu(\bm{z} - \bm{y})| d\bm{z}.
	\end{align}
	
	We prove this claim as follows. Without loss of generality, we take $\bm{x} = 0$ for this step of the proof (via a linear change of variables). Let $g(\bm{x})$ be a function that is defined as $\nabla G(\bm{x})$ where $G$ is differentiable. At the (measure 0) set of points where $G$ is not differentiable, we define $g$ to be equal to $\nabla F_i(\bm{x})$ for an arbitrary $i \in \arg\min_{j} F_j(\bm{x})$. With probability 1, $\E[g(\bm{x} + \bm{z})] = \E[\nabla G(\bm{x} + \bm{z})]$ We have
	
	\begin{align*}
	\E[g(\bm{z}) - g(\bm{\bm{y}} + \bm{z})] &= \int g(\bm{z}) \mu(\bm{z}) dz - \int g(\bm{y} + \bm{z}) \mu(\bm{z}) dz\\
	&= \int g(\bm{z}) \mu(\bm{z}) dz - \int g(\bm{z})\mu(\bm{z} - \bm{y})dz\\
	&= \int_{I_>} g(\bm{z}) \left[\mu(\bm{z}) - \mu(\bm{z} - \bm{y})\right] + \int_{I_<} g(\bm{z}) \left[\mu(\bm{z}-\bm{y}) - \mu(\bm{z})\right]
	\end{align*}
	
	where $I_> = \{\bm{z} | \mu(\bm{z}) > \mu(\bm{z} - \bm{y})\}$ and $I_< = \{\bm{z} | \mu(\bm{z}) < \mu(\bm{z} - \bm{y})\}$. Taking norms, we have
	
	\begin{align*}
	||\E[g(\bm{\bm{z}}) - g(\bm{y} + \bm{\bm{z}})]||_\infty \,&\leq \sup_{\bm{z}\in I_> \cup I_<} ||g(\bm{z})||_\infty \left|\int_{I_>} \left[\mu(\bm{z}) - \mu(\bm{z} - \bm{y})\right] + \int_{I_<} \left[\mu(\bm{z}-\bm{y}) - \mu(\bm{z})\right]\right|\\
	&\leq M \left|\int_{I_>} \left[\mu(\bm{z}) - \mu(\bm{z} - \bm{y})\right] + \int_{I_<} \left[\mu(\bm{z}-\bm{y}) - \mu(\bm{z})\right]\right|\\
	&= M\int |\mu(\bm{z}) - \mu(\bm{z} - \bm{y})| dz
	\end{align*}
	
	Having proved that Equation \ref{eq:duchi-claim} holds, we now just need to show $\int | \mu(\bm{z} - \bm{x}) - \mu(\bm{z} - \bm{y})| d\bm{z} \leq \frac{||\bm{x}-\bm{y}||_1}{u}$. This follows from Duchi et al.\ (2012), Lemma 12.
	
\end{proof}

We now prove a technical smoothness lemma. The argument is standard, but we include it for completeness.

\begin{lemma}
	For any $\bm{x}, \bm{y}$, $G_\mu(\bm{x} + \gamma \bm{y}) - G_\mu(\bm{x}) \geq \gamma\nabla G_\mu(\bm{x})^T \bm{y} - \frac{M k^2 \gamma^2}{2u}$. \label{lemma:linearization}
\end{lemma}

\begin{proof}
	For any $\bm{x}, \bm{y} \in \mathcal{P}$, we consider the one dimensional auxiliary function $g_{\bm{x},\bm{y}}(\xi) = G_\mu(\bm{x} + \xi \bm{y})$. We have
	%
	%
	%
	%
	
	\begin{align*}
	G_\mu(\bm{x} + \gamma \bm{y}) -  G_\mu(\bm{x}) &= \int_{\xi = 0}^1 \frac{d g_{\bm{x},\gamma\bm{y}}(\xi)}{d \xi} d\xi\\
	&= \int_{\xi = 0}^1 \nabla G_\mu(\bm{x} + \xi \gamma \bm{y})^\top (\gamma \bm{y}) d\xi \\
	&= \gamma \int_{\xi = 0}^1 \nabla G_\mu(\bm{x})^\top \bm{y} + \left[\nabla G_\mu(\bm{x} + \xi \gamma \bm{y})^\top - \nabla G_\mu(\bm{x})^\top \right]\bm{y}  d\xi\\
	&\geq \gamma \int_{\xi = 0}^1 \nabla G_\mu(\bm{x})^\top \bm{y} - ||\nabla G_\mu(\bm{x} + \xi \gamma \bm{y})^\top - \nabla G_\mu(\bm{x})^\top||_\infty ||\bm{y}||_1  d\xi \text{ (by H\"{o}lder's inequality)}\\
	&\geq \gamma \int_{\xi = 0}^1 \nabla G_\mu(\bm{x})^\top \bm{y} - \frac{M}{\mu}||\xi \gamma \bm{y}||_1 ||\bm{y}||_1  d\xi \text{ $(\nabla G_\mu$ is $\frac{M}{\mu})$-Lipschitz}\\ 
	&\geq \gamma \nabla G_\mu(\bm{x})^\top - \gamma^2 \int_{\xi = 0}^1 \frac{M k^2}{u} \xi d\xi \text{ (bound on $\ell_1$ diameter of $\mathcal{P}$)}\\
	&= \gamma \nabla G_\mu(\bm{x})^\top - \frac{\gamma^2 M k^2}{2u}
	\end{align*}
	which proves the lemma. 
\end{proof}

We also use the following lemma, the proof of which can be found in Bian et al.\ (2017):

\begin{lemma}
	For any DR-submodular function $G$ and its optimizer $\bm{x}^*$, $G(\bm{x}^* + \bm{x}) - G(\bm{x}) \leq \nabla G(\bm{x})^\top \bm{x}^*$. \label{lemma:subbound}
\end{lemma}

We can now proceed to prove our guarantee on the performance of the SFW algorithm for optimizing the objective $G$.

\begin{proof}[Proof of Theorem \ref{theorem:sfw}]
	We analyze the gain made in a single step of SFW as follows:
	
	\begin{align*}
	&G_\mu(\bm{x}^\ell) - G_\mu(\bm{x}^{\ell-1}) \geq \gamma_\ell \nabla G_\mu(\bm{x}^{\ell-1})^\top \bm{v}^\ell - \frac{Mk^2}{2u} \gamma_\ell^2 \text{ (Lemma \ref{lemma:linearization})}\\
	&= \gamma_\ell \tilde{\nabla}_\ell^\top \bm{v}^\ell - \gamma_\ell \left(\tilde{\nabla}_\ell - \nabla G_\mu(\bm{x}^{\ell-1})\right)^\top \bm{v}^\ell - \frac{Mk^2}{2u} \gamma_\ell^2\\
	&\geq \gamma_\ell \tilde{\nabla}_\ell^\top \bm{v}^\ell - \gamma_\ell k ||\tilde{\nabla}_\ell - \nabla G_\mu(\bm{x}^{\ell-1})||_\infty  - \frac{Mk^2}{2u} \gamma_\ell^2 \text{ (by H\"{o}lder's inequality and $rank(\mathcal{M}) = k$)}\\
	&\geq \gamma_\ell \tilde{\nabla}_\ell^\top \bm{x}^* - \gamma_\ell k ||\tilde{\nabla}_\ell - \nabla G_\mu(\bm{x}^{\ell-1})||_\infty  - \frac{Mk^2}{2u} \gamma_\ell^2 \text{ (by definition of $\bm{x}^*$)}\\
	&= \gamma_\ell \nabla G_\mu(\bm{x}^{\ell-1})^\top \bm{x}^* - \gamma_\ell \left(G_\mu(\bm{x}^{\ell-1}) - \tilde{\nabla}_\ell\right)^\top \bm{x}^* - \gamma_\ell k ||\tilde{\nabla}_\ell - \nabla G_\mu(\bm{x}^{\ell-1})||_\infty  - \frac{Mk^2}{2u} \gamma_\ell^2 \\
	&\geq \gamma_\ell \nabla G_\mu(\bm{x}^{\ell-1})^\top \bm{x}^* - 2\gamma_\ell k ||\tilde{\nabla}_\ell - \nabla G_\mu(\bm{x}^{\ell-1})||_\infty  - \frac{Mk^2}{2u} \gamma_\ell^2 \\
	&\geq \gamma_\ell \left(G_\mu (\bm{x}^* + \bm{x}^{\ell-1}) - G_\mu (\bm{x}^{\ell-1}) \right)- 2\gamma_\ell k||\tilde{\nabla}_\ell - \nabla G_\mu(\bm{x}^{\ell-1})||_\infty  - \frac{Mk^2}{2u} \gamma_\ell^2 \text{( Lemma \ref{lemma:subbound} and $\bm{x}^* \succeq 0$)}\\
	&\geq \gamma_\ell \left(G_\mu (\bm{x}^*) - G_\mu (\bm{x}^{\ell-1}) \right)- 2\gamma_\ell k||\tilde{\nabla}_\ell - \nabla G_\mu(\bm{x}^{\ell-1})||_\infty - \frac{Mk^2}{2u} \gamma_\ell^2 \text{( monotonicity)}\\
	\end{align*}
	
	Now we give a high probability bound on $||\tilde{\nabla}_\ell - \nabla G_\mu(\bm{x}^{\ell-1})||_\infty$. Denote by $\tilde{\nabla}_\ell^i$ the $i$th randomly sampled gradient and $\left[\tilde{\nabla}_\ell^i\right]_j$ its $j$th entry (the derivative with respect to item $j$). We will give a high probability bound on each individual entry of the estimated gradient and them combine them using union bound to control $||\tilde{\nabla}_\ell - \nabla G_\mu(\bm{x}^{\ell-1})||_\infty$. 
	
	Fix any individual entry $j$. We have $\left[\tilde{\nabla}_\ell\right]_j = \frac{1}{c}\sum_{i = 1}^c \left[\tilde{\nabla}_\ell^i\right]_j$. Because the first-order oracle returns an unbiased estimate, we know that $\E\left[\left[\tilde{\nabla}_\ell\right]_j - \nabla_j G_\mu(\bm{x}^{\ell-1})\right] = 0$. Further, $\left|\left[\tilde{\nabla}_\ell\right]_j\right| \leq M$ and $\left|\nabla_j G_\mu(\bm{x}^{\ell-1})\right| \leq M$, so $\left|\left[\tilde{\nabla}_\ell\right]_j - \nabla_j G_\mu(\bm{x}^{\ell-1})\right| \leq 2M$ holds via triangle inequality. Now via Hoeffding's inequality, we have that 
	
	\begin{align*}
	\text{Pr}\left[\left|\sum_{i =1}^c\left[\tilde{\nabla}_\ell^i\right]_j - c\nabla_j G_\mu(\bm{x}^{\ell-1})\right| \geq m\frac{\epsilon}{8k}\right] \leq 2e^{-\frac{\epsilon^2 c}{128 k^2 M^2}}
	\end{align*}
	
	and so taking $c = \frac{128 M^2 k^2}{\epsilon^2} \log \frac{4 K n}{\delta}$ ensures that 
	
	\begin{align*}
	\text{Pr}\left[\left|\left[\tilde{\nabla}_\ell\right]_j - \nabla_j G_\mu(\bm{x}^{\ell-1})\right| \geq \frac{\epsilon}{8k}\right] \leq \frac{\delta}{2 K n}.
	\end{align*}
	
	By union bound, the total probability of this event holding for all $n$ items at each of the $K$ timesteps is at least $1- \frac{\delta}{2}$. In all of what follows, we condition on this happening. Rearranging gives
	
	\begin{align*}
	G_\mu(\bm{x}^*) - G_\mu(\bm{x}^\ell) &\leq (1 - \gamma_\ell)\left[G_\mu(\bm{x}^*) - G_\mu(\bm{x}^{\ell-1})\right] - 2\gamma_\ell k ||\tilde{\nabla}_\ell - \nabla G_\mu(\bm{x}^{\ell-1})||_\infty - \frac{Mk^2}{2u} \gamma_\ell^2\\
	&\leq (1 - \gamma_\ell)\left[G_\mu(\bm{x}^*) - G_\mu(\bm{x}^{\ell-1})\right] - \frac{\gamma_\ell \epsilon}{4} - \frac{Mk^2}{2u} \gamma_\ell^2
	\end{align*}
	
	and so after $K$ iterations we obtain
	
	\begin{align*}
	G_\mu(\bm{x}^*) - G_\mu(\bm{x}^K) &\leq \prod_{\ell = 0}^{K-1}(1 - \gamma_\ell)\left[G_\mu(\bm{x}^*) - G_\mu(\bm{x}^{0})\right] - \sum_{\ell =0}^{K-1}\frac{\gamma_\ell \epsilon}{4} - \sum_{\ell =0}^{K-1}\frac{Mk^2}{2u} \gamma_\ell^2\\
	&\leq e^{-\sum_{\ell = 0}^{K-1} \gamma_\ell} G_\mu(\bm{x}^*) - \sum_{\ell =0}^{K-1}\frac{\gamma_\ell \epsilon}{4} - \sum_{\ell =0}^{K-1}\frac{Mk^2}{2u} \gamma_\ell^2
	\end{align*}
	
	with constant stepsize $\gamma = \frac{1}{K}$, we have $\sum_{\ell  =0}^{K-1} \gamma_\ell = 1$ and $\sum_{\ell  =0}^{K-1} \gamma_\ell^2 = \frac{1}{K}$, this yields
	
	\begin{align*}
	G_\mu(\bm{x}^*) - G_\mu(\bm{x}^K) &\leq \frac{1}{e} G_\mu(\bm{x}^*) - \frac{\epsilon}{4} - \frac{Mk^2}{2u K}
	\end{align*}
	
	and hence
	
	\begin{align*}
	G(\bm{x}^*) - G(\bm{x}^K) &\leq \frac{1}{e} G_\mu(\bm{x}^*) -\frac{\epsilon}{4} - \frac{Mk^2}{2u K} - M n u.
	\end{align*}
	
	and so taking $u = \frac{\epsilon}{4Mn}$ and $K = \frac{8 M^2 k^2 n}{\epsilon^2}$ ensures that $G(\bm{x}^*) - G(\bm{x}^K) \leq \frac{1}{e} G_\mu(\bm{x}^*) -\frac{3\epsilon}{4}$

	Now we just need a small trick to deal with the issue that $G$ is only defined for $\bm{x} \in [0,1]^n$, and random perturbation by $\bm{z}$ may take us out of this range. To avoid negative values, we start the algorithm at the point $u \bm{1}$. Since each coordinate only increases during the course of the algorithm, we are guaranteed to query $G$ only at nonnegative points. To deal with values greater than 1, we can instead analyze the function $H(\bm{x}) = G(\bm{x} \land \bm{1})$, where $\land$ denotes coordinate-wise maximum. $H$ is also up-concave, and agrees with $G$ at every point in $\polym$. After running SFW applied to $H$ for $K$ iterations, we obtain via Theorem \ref{theorem:sfw} a solution $\bm{x}^K$ such that $H(\bm{x}^K) \geq (1 - \frac{1}{e}) \max_{\bm{x} \in \polym} H(\bm{x}) - \epsilon = (1 - \frac{1}{e}) \max_{\bm{x} \in \polym} G(\bm{x}) - \epsilon$. The issue is that we may have $\bm{x}^K \not\in \polym$. We output the point $\bm{x}^K - u\bm{1}$, which is guaranteed to lie in $\polym$. To analyze the loss incurred, we use the following lemma
	
	\begin{lemma} \label{lemma:submod-prob}
		Let $f$ be a monotone submodular function with $\max_i f(\{i\}) \leq M$. Let $R(\bm{x})$ be a random set in which every element appears independently with probability $x_i \geq u$. Then $\E[f(R(\bm{x} - u\bm{1}))] \geq \E[f(R(\bm{x}))] - u M n$.
	\end{lemma}
	
	\begin{proof}
		We decompose the expected value of $R(\bm{x} - u\bm{1})$ into the expected marginal contribution from each item:
		\allowdisplaybreaks
		\begin{align*}
		\E[f(R(\bm{x} - u\bm{1}))] &= \sum_{i = 1}^n \text{Pr}[i \in f(R(\bm{x} - u\bm{1}))]\E[f(i|(R(\bm{x} - u\bm{1}))] &&\text{(linearity of expectation)}\\
		&\geq \sum_{i = 1}^n \text{Pr}[i \in f(R(\bm{x} - u\bm{1}))]\E[f(i|(R(\bm{x}))] && \text{(submodularity)}\\
		&= \sum_{i = 1}^n (x_i - u)\E[f(i|(R(\bm{x}))]\\
		&= \sum_{i = 1}^n x_i \E[f(i|(R(\bm{x}))] - u \sum_{i = 1}^n \E[f(i|(R(\bm{x}))]\\
		&\geq \sum_{i = 1}^n x_i \E[f(i|(R(\bm{x}))] - u \sum_{i = 1}^n \E[f(\{i\})] && \text{(submodularity)}\\
		&\geq \sum_{i = 1}^n x_i \E[f(i|(R(\bm{x}))] - u n M && \text{($f(\{i\}) \leq M$)}\\
		&= \E[f(R(\bm{x}))] - u M n.
		\end{align*}
	\end{proof}
	
	Applying Lemma \ref{lemma:submod-prob} to every $f_i \in \mathcal{F}$, we conclude that 
	\begin{align*}
	H(\bm{x}^K - u\bm{1}) \geq H(\bm{x}^K) - u M n = G(\bm{x}^K) - \frac{\epsilon}{4} \geq \left(1 - \frac{1}{e}\right)G(\bm{x}^*) - \epsilon 
	\end{align*}
	
	which completes the proof. 
\end{proof}

Lastly, we prove our concentration guarantee for the output of the swap rounding algorithm.

\begin{proof}[Proof of Theorem \ref{theorem:sample}]
	For now, fix a specific function $f_i$. We will show that with high probability, the expected value of $f_i$ on the empirical distribution is close to its expected value on the full distribution induced by randomized swap rounding. At the end we will take a union bound over all $f_i \in \mathcal{F}$. Let $R_\ell$ be the random set drawn in the $\ell th$ iteration of randomized swap rounding. Let $\mu_0 = F_i(\bm{x}^K)$.  Note that for all $\ell$, $\E[f_i(R_\ell)] \geq \mu_0$ via the guarantee for randomized swap rounding. Let $Y = \sum_{\ell = 1}^r f_i(R_\ell)$ and note that $\E[Y] \geq r \mu_0$. 
	
	Our high-level approach is to apply Markov's inequality to the random variable $e^{r \mu_0 - Y}$. Let $\lambda$ be an arbitrary parameter in $[0,1]$ (later, we will set $\lambda$ to get the best bound). For any single iteration of randomized swap rounding, Chekuri et al.\ bound the exponential moment of the random variable $\lambda(\mu_0 - f_i(R_\ell))$ as
	
	\begin{align*}
	\E[e^{\lambda (\mu_0 - f(R_\ell))}]  \leq e^{2 \lambda^2 \mu_0}.
	\end{align*}  
	
	Since $r \mu_0 - Y  = \sum_{\ell = 1}^r \mu_0 - f(R_\ell)$, we have
	
	\begin{align*}
	\E\left[e^{\lambda (r \mu_0 - Y)}\right] &= \E\left[e^{ \sum_{\ell = 1}^r \lambda(\mu_0 - f(R_\ell))}\right]\\
	&= \E\left[\prod_{\ell=1}^r e^{\lambda(\mu_0 - f(R_\ell))}\right]\\
	&= \prod_{\ell=1}^r \E\left[e^{\lambda(\mu_0 - f(R_\ell))}\right] \quad\text{ (independence)}\\
	&\leq e^{2 r \lambda^2 \mu_0}.
	\end{align*}
	
	Applying Markov's inequality yields
	
	\begin{align*}
	\text{Pr}[r\mu_0 - Y \geq \epsilon r \mu_0] &= \text{Pr}\left[e^{\lambda(r\mu_0 - Y )} \geq e^{\epsilon r \lambda \mu_0}\right] \\
	&\leq \frac{\E\left[e^{\lambda(r\mu_0 - Y )} \right]}{e^{\epsilon r \lambda \mu_0}}\\
	&\leq e^{2r\lambda^2 \mu_0 - \epsilon r \lambda \mu_0}
	\end{align*}
	
	Taking $\lambda = \frac{\epsilon}{4}$, we obtain
	
	\begin{align*}
	\text{Pr}\left[\frac{1}{r}\sum_{\ell=1}^r f_i(R_\ell) \leq (1 - \epsilon)\mu_0\right] = \text{Pr}[r\mu_0 - Y \geq \epsilon r \mu_0] \leq  e^{\frac{- r \mu_0 \epsilon^2}{8}}.
	\end{align*}
	
	We now distinguish two cases. First, $\mu_0 < \epsilon$. Since $f_i(R_\ell) \geq 0 \,\, \forall i, \ell$, $\frac{1}{r}\sum_{\ell = 1}^r f_i(R_\ell) \geq \mu_0 - \epsilon$ holds with probability 1. Second, $\mu_0 \geq \epsilon$. Here, we see that setting $r = \Theta\left(\frac{1}{\epsilon^3} \left(\log |\mathcal{F}| + \log \frac{1}{\delta}\right)\right)$ ensures that $\frac{1}{r}\sum_{\ell = 1}^r f_i(R_\ell) \geq (1 - \epsilon)\mu_0$ holds with probability at least $1 - \frac{\delta}{|\mathcal{F}|}$. Taking union bound over all $f_i \in \mathcal{F}$ completes the proof. 
	
\end{proof}

\end{document}